\documentclass[11pt]{article}
\usepackage{palatino}
\usepackage{pgfplots} 
\usepackage{tikz, tikzpeople}
\usetikzlibrary{calc,intersections,decorations.pathmorphing,decorations.markings,patterns,arrows.meta,positioning}
\usetikzlibrary{angles,quotes}
\usetikzlibrary{arrows.meta}
\usetikzlibrary{arrows}
\usetikzlibrary{shapes.multipart}
\usetikzlibrary{fit}
\usetikzlibrary{shapes.geometric,positioning}
\usetikzlibrary{calc}
\usepackage{bm}
\usepackage{mathtools}
\usepackage{xfrac}
\usepackage{graphicx}
\usepackage[linesnumbered,ruled,vlined]{algorithm2e}
\usepackage{cite}
\usepackage{amsmath,amssymb,amsfonts, amsthm,  mathrsfs}
\usepackage[linesnumbered,ruled,vlined]{algorithm2e}
\usepackage{algpseudocode}
\usepackage{textcomp}
\usepackage{xcolor}
\usepackage{braket}
\usepackage{enumitem}
\usepackage{cancel}
\usepackage{hyperref}
\usepackage{stfloats}
\usepackage{xparse}

\newcommand{\cmark}{\ensuremath{\checkmark}}
\newcommand{\xmark}{\ensuremath{\times}}

\usetikzlibrary{patterns}
\newtheorem{lemma}{Lemma}
\newtheorem{theorem}{Theorem}

\newtheorem{definition}{Definition}

\def\bbsmatrix#1{\begin{bsmallmatrix}#1\end{bsmallmatrix}}

\newlength{\leftstackrelawd}
\newlength{\leftstackrelbwd}
\def\leftstackrel#1#2{\settowidth{\leftstackrelawd}%
{${{}^{#1}}$}\settowidth{\leftstackrelbwd}{$#2$}%
\addtolength{\leftstackrelawd}{-\leftstackrelbwd}%
\leavevmode\ifthenelse{\lengthtest{\leftstackrelawd>0pt}}%
{\kern-.5\leftstackrelawd}{}\mathrel{\mathop{#2}\limits^{#1}}}

\newcommand{\mN}{\mathsf{N}} 
\newcommand{\mZ}{\mathsf{Z}} 

\NewDocumentCommand{\mP}{e{_}}{%
  \mathsf{P}%
  \IfValueT{#1}{_{\!#1}}%
}

\newcommand{\mQ}{\mathsf{Q}}

\newcommand{\mB}{\mathsf{B}}

\newcommand{\NS}{\mathchoice
  {\mathrm{\scriptscriptstyle NS}} 
  {\mathrm{\scriptscriptstyle NS}} 
  {\mathrm{\scriptscriptstyle NS}} 
  {\mathrm{\scriptscriptstyle NS}} 
}

\newcommand{\TxNS}{\mathchoice
  {\mathrm{\scriptscriptstyle Tx\!-\!NS}} 
  {\mathrm{\scriptscriptstyle Tx\!-\!NS}} 
  {\mathrm{\scriptscriptstyle Tx\!-\!NS}} 
  {\mathrm{\scriptscriptstyle Tx\!-\!NS}} 
}

\allowdisplaybreaks

\title{The Capacity Region of the Multiple Access Channel with Non-Signaling Assistance}
\author{Yuhang Yao, Syed A. Jafar\\
{\small Nhu Department of Electrical Engineering and Computer Science}\\
{\small University of California Irvine, Irvine, CA 92697}\\
{\small \it Email: \{yuhangy5, syed\}@uci.edu}
}

\date{}

\begin{document}
\maketitle

\begin{abstract}
The capacity region of the $K$-sender discrete memoryless multiple access channel (MAC) is fully characterized when non-signaling (NS) assistance is available to all $K$ transmitters and the receiver. It is shown to have the same form as the classical capacity region of the MAC, except that the input distribution is allowed to be arbitrarily dependent across the senders. In particular, the NS-assisted capacity region matches the natural generalization to $K$ senders of an outer bound that was previously established by Fawzi and Ferm\'{e} for $K=2$ senders. Additionally, we provide examples of $K$-sender MACs where the multiplicative gain in capacity from NS-assistance is arbitrarily close to $K$. Combined with an upper bound from prior work, this establishes $K$ as the extremal value of the multiplicative gain from NS-assistance across all $K$-sender MAC settings.
\end{abstract}

\section{Introduction} 
With quantum technologies on the horizon there is the need to expand the scope of classical network information theory to allow superclassical resources, such as non-local correlations. A non-signaling (NS) extension of classical information theory retains its focus on classical communication but allows communicating nodes to share arbitrary non-signaling correlations in advance --- i.e., those correlations which cannot by themselves constitute a communication channel across any network partition. NS-correlations form a natural superset of quantum correlations, and are often much more tractable  as they can be described as a polytope with finitely many linear constraints, whereas a characterization of quantum correlations  generally involves infinite hierarchies of nonlinear operator-algebraic relations along with semidefinite constraints \cite{Nonlocality}. The NS-assisted capacity of a classical communication network thus provides arguably the tightest seemingly tractable relaxation of classical capacity which cannot be surpassed with the help of any quantum entanglement that is shared in advance between the classical nodes. 

The framework of NS-assisted capacity of classical communication networks is still in its early stages. However, there has been some recent progress which we summarize next. 

The NS-assisted capacity of a point-to-point channel is shown by Matthews in \cite{matthews2012linear} to be the same as its classical capacity, strengthening the prior result of Bennett et al. \cite{Bennett_Shor_Smolin_Thapliyal_PRL}, which showed that quantum correlations do not increase the capacity of a classical point-to-point channel. For the point-to-point \emph{channel with state} where the channel state information at the transmitter (CSIT) is available non-causally, the capacity is shown in \cite{Yao_Jafar_CSITTP} to be the same as the capacity of a corresponding classical setting where the channel state is also made available to the receiver --- a phenomenon that \cite{Yao_Jafar_CSITTP} labels as \emph{virtual signaling of CSIT}. Remarkably with NS-assistance, \cite{Yao_Jafar_CSITTP} finds that not only the capacity but also the optimal probability of success remains unchanged if the channel state is additionally provided to the receiver. Reference \cite{Yao_Jafar_NSCWS} shows that the NS-assisted capacity of a point-to-point channel remains the same as if non-causal CSIT was available, even if the CSIT is actually made available only \emph{causally}, albeit the aforementioned stronger equivalence in probability of success no longer holds in the causal setting. For both causal and non-causal CSIT, it is noted in \cite{Yao_Jafar_CSITTP, Yao_Jafar_NSCWS} that there exist point-to-point channels where the multiplicative gain in capacity from NS-assistance can be arbitrarily large. Building on these NS-assisted capacity gains,  \cite{Yao_Jafar_QEACWS} identifies channels with causal CSIT where significant gains over classical capacity are achievable with quantum-entanglement assistance.

Beyond the point-to-point channel, NS-assisted capacity has been explored in multiple access channel (MAC), broadcast channel (BC), and interference channel (IC) settings. Quek and Shor in \cite{Quek_Shor,Quek_Dissertation} show that NS or quantum-assistance can improve the capacity of interference and multiple-access channels, leaving open the corresponding question for the ``\emph{third pillar of network information theory}" --- the broadcast channel. The existence of NS-assisted capacity gains in a BC is identified as a `major open problem' by Fawzi and Ferm\'{e}  who study it in \cite{fawzi2024broadcast} from an algorithmic perspective. The question is resolved in \cite{Yao_Jafar_NS_DoF} where the extremal value of the multiplicative gain from NS-assistance across all $K$-user BC settings is shown to be $K$. More recently, the NS-assisted capacity region of the discrete memoryless $K$-user BC is fully characterized in \cite{Yao_Jafar_NSSato}. Remarkably, the NS-assisted capacity region of the BC coincides with the so-called \emph{Sato region}.

The MAC and the BC are often considered as two sides of the same coin, with the MAC often being the more tractable of the two. Indeed, many of the earliest results on NS-assisted capacity in multiuser settings appear in the context of the MAC. Building on the findings of Quek and Shor in \cite{Quek_Shor,Quek_Dissertation}, Leditzky et al. \cite{leditzky2020playing} show that a two-sender MAC constructed from a nonlocal game can exhibit a capacity improvement when the two transmitters share quantum entanglement. Subsequent works, including \cite{seshadri2023separation, QMAC_Yun, Qcoop_Nam}, further study MACs constructed from nonlocal games and develop bounds on their correlation-assisted sum-rate capacities. Pereg et al. \cite{pereg2024MAC_QEassist} study general classical two-sender MACs with quantum entanglement shared between the transmitters and establish inner and outer bounds on the corresponding capacity region. More recently, building on the quantum-assisted capacity gains in the point-to-point setting with causal CSIT \cite{Yao_Jafar_NSCWS}, it is shown in \cite{Yao_Jafar_Unlock} that quantum-assistance allows exponential/unbounded multiplicative capacity gains in a MAC with causal CSIT. The channel with state model  is critical to such gains. Indeed,  \cite{Yao_Jafar_Unlock} showed that for a $K$-user MAC without state, the extremal value of multiplicative capacity gain, even with NS-assistance, cannot exceed $K$. 

Let us focus henceforth on the capacity of the MAC without state, when NS-assistance is available to all parties. This brings us to the work by Fawzi and Ferm\'{e} in \cite{fawzi2024MAC} which most directly motivates the present paper. Exploring the NS-assisted capacity of the MAC without state, Fawzi and Ferm\'{e} raise the following questions: ``\emph{Can NS correlations lead to significant gains in capacity for natural MACs? Can we find a characterization of the capacity region of the MAC when NS resources between the parties are allowed?}" In particular, Fawzi and Ferm\'{e} obtain an outer bound on the NS-assisted capacity region by introducing a certain relaxed-NS-assistance model. While they show that the relaxed-NS-assistance model can be loose in terms of the optimal probability of success, they leave open the possibility that the bound could be tight for the capacity region, identifying it as ``\emph{the main open problem}" in \cite{fawzi2024MAC}. This open problem is the motivation for the present work.

As our main result, we fully characterize (Theorem \ref{thm:NSMAC} in Section \ref{sec:Kcapregion}) the capacity region of the $K$-sender discrete memoryless MAC when NS-assistance is available to all $K$ transmitters and the receiver. The capacity region is shown to have the same form as the classical capacity region of the MAC, except that the input distribution is allowed to be arbitrarily dependent across the senders. In particular, the NS-assisted capacity region matches the natural generalization to $K$ senders of the outer bound established by Fawzi and Ferm\'{e} for $K=2$ senders. Additionally, we provide examples of $K$-sender MACs (Theorem \ref{thm:sum_prod_lock} in Section \ref{sec:extremal}) where the multiplicative gain in capacity from NS-assistance is arbitrarily close to $K$. Combined with the previously noted upper bound from \cite{Yao_Jafar_Unlock}, this establishes $K$ as the extremal value of the multiplicative gain from NS-assistance across all $K$-sender MAC settings.

\section{Preliminaries}
\subsection{Sets}
Let $\mathbb{N}$ denote the set of positive integers. Let $\mathbb{R}_{\geq 0}$ and $\mathbb{R}_{> 0}$ denote the sets of non-negative and positive real numbers, respectively. 
For $n\in \mathbb{N}$, the notation $[n]$ denotes the set of integers $\{1,2,\ldots, n\}$, and $x^n$ denotes $(x_1,x_2,\ldots, x_n)$.
We write $\mathcal{S} \subseteq \mathcal{T}$ if $\mathcal{S}$ is a subset of $\mathcal{T}$, and $\mathcal{S} \subset \mathcal{T}$ if $\mathcal{S}$ is a proper subset of $\mathcal{T}$.
For $\mathcal{S}=\{i_1,i_2,\ldots,i_m\}\subset\mathbb{N}$, the notation $(A_i)_{i\in\mathcal{S}}$ denotes the tuple $(A_{i_1},A_{i_2},\ldots,A_{i_m})$, where $i_1<i_2<\cdots<i_m$. 
We use $\mathbb{I}(\cdot)$ to denote the indicator function, which equals $1$ if the predicate is true and $0$ otherwise.

\subsection{Probability}
A probability mass function (pmf) on a finite-cardinality discrete set $\mathcal{X}$ is a function $\mP_{X}\colon \mathcal{X} \to \mathbb{R}_{\geq 0}$ satisfying $\sum_{x\in \mathcal{X}} \mP_X(x) = 1$. The set of all pmfs on $\mathcal{X}$ is denoted by $\mathcal{P}(\mathcal{X})$. 
The notation $X \sim \mP$ indicates that the random variable $X$ is distributed according to the pmf $\mP$.
Let $\mathcal{X} \times \mathcal{Y}$ denote the Cartesian product of $\mathcal{X}$ and $\mathcal{Y}$, and let $\mathcal{X}^n$ denote the $n$-fold Cartesian product of $\mathcal{X}$. 
A joint pmf $\mP_{X_1,\ldots,X_n}$ is simply a pmf on $\mathcal{X}_1\times\cdots\times\mathcal{X}_n$.
A conditional pmf $\mP_{X\mid Y}$ associated with an ordered pair of discrete sets $(\mathcal{X},\mathcal{Y})$ is a function $\mP_{X\mid Y} \colon \mathcal{X}\times \mathcal{Y} \to \mathbb{R}_{\geq 0}$ such that, for every $y\in \mathcal{Y}$, $\sum_{x}\mP_{X\mid Y}(x\mid y) = 1$. The set of all conditional pmfs associated with $(\mathcal{X}, \mathcal{Y})$ is denoted by $\mathcal{P}(\mathcal{X} \mid \mathcal{Y})$. 
Given a joint pmf $\mP_{XY}$, the marginal pmfs $\mP_X$ and $\mP_Y$, as well as the conditional pmfs $\mP_{Y\mid X}$ and $\mP_{X\mid Y}$, are defined in the usual way. Whenever $\mP_X(x)=0$, the conditional pmf $\mP_{Y\mid X=x}$ may be chosen arbitrarily from $\mathcal{P}(\mathcal{Y})$.
The notation $\mP_X^{\otimes n}(x^n)$ denotes the product pmf $\mP_X^{\otimes n}(x^n)=\prod_{i=1}^n \mP_{X}(x_i)$, and $\mP_{Y\mid X}^{\otimes n}(y^n \mid x^n) = \prod_{i=1}^n \mP_{Y\mid X}(y_i\mid x_i)$.
We use $\Pr(E)$ to denote the probability of a random event $E$, and $\mathbb{E}[X]$ to denote the expectation of $X$.

\subsection{Information measures and typicality}
For $X\sim \mP_X$, let $H(X) = -\sum_{x} \mP_X(x) \log_2 \mP_X(x) = -\mathbb{E}_{X\sim \mP_X}[\log_2(\mP_X(X))]$ denote the entropy of $X$. For $(X,Y)\sim \mP_{XY}$, let $I(X;Y) = \sum_{(x,y)}\mP_{XY}(x,y) \log_2 \frac{\mP_{XY}(x,y)}{\mP_{X}(x)\mP_Y(y)}$ denote the mutual information of $X,Y$, and let $I(X;Y\mid Z) = \sum_{(x,y,z)}\mP_{XYZ}(x,y,z) \log_2 \frac{\mP_{XY\mid Z}(x,y\mid z)}{\mP_{X\mid Z}(x\mid z)\mP_{Y\mid Z}(y\mid z)}$ denote the conditional mutual information of $X,Y$ given $Z$.
The subscript indicating the underlying distribution may be omitted when it is clear from the context. 
 
For a sequence $x^n\in \mathcal{X}^n$, its empirical pmf (type) $\widehat{\mP}_{x^n}\in \mathcal{P}(\mathcal{X})$ is defined by
\begin{align} \label{eq:def_type}
	\widehat{\mP}_{x^n}(x) \triangleq \frac{1}{n} \sum_{i=1}^n \mathbb{I}(x_i=x), \qquad \forall x\in \mathcal{X}.
\end{align} 
Given $\mP_X\in \mathcal{P}(\mathcal{X})$, $\epsilon \in (0,1)$ and $n\in \mathbb{N}$, define the (strongly) typical set
\begin{align} \label{eq:def_typical_set}
	\mathcal{T}_{\epsilon}^{(n)}(\mP_X) \triangleq \big\{ x^n \in \mathcal{X}^n \colon |\widehat{\mP}_{x^n}(x)-\mP_X(x)| \leq \epsilon \mP_X(x), \forall x\in \mathcal{X} \big\}.
\end{align}
For jointly distributed random variables $(X_1,\ldots, X_m)$, the jointly typical set is defined by treating $(X_1,\ldots, X_m)$ as a single random variable.

\section{Problem formulation}
\subsection{Non-signaling assisted $K$-sender discrete memoryless multiple access channel}
A $K$-sender discrete memoryless multiple access channel (DM-MAC) is specified by a conditional pmf $\mN_{Y\mid X_1X_2\cdots X_K} \in \mathcal{P}(\mathcal{Y} \mid \mathcal{X}_1 \times \mathcal{X}_2 \times\cdots \times \mathcal{X}_K)$ with underlying finite-cardinality discrete sets $\mathcal{Y},\mathcal{X}_1,\mathcal{X}_2,\ldots, \mathcal{X}_K$. For each use of the channel, Transmitter $k$ inputs $x_k \in \mathcal{X}_k$ for every $k\in [K]$, and the probability that the receiver receives $Y=y \in \mathcal{Y}$ conditioned on the channel inputs $X_k=x_k\in \mathcal{X}_k, \forall k\in [K]$, is equal to $\mN_{Y\mid X_1X_2\cdots X_K}(y\mid x_1,x_2,\ldots, x_K)$.
\begin{figure}[htbp]
\centering
\begin{tikzpicture}[scale=0.85, transform shape]

    \node (Z) [
        rectangle,
        draw,
        thick,
        minimum width=13cm,
        minimum height=2cm
    ] at (0,0) {};

    \draw[thick] (-4.5,1) -- (-4.5,-1);
    \draw[thick] (-2.5,1) -- (-2.5,-1);
    \draw[thick] (-0.5,1) -- (-0.5,-1);
    \draw[thick] (1.5,1) -- (1.5,-1);

    \node [
        rectangle,
        minimum height=0.8cm,
        minimum width = 10cm,
        fill=gray!30
    ] at (0,0) {
   
        $\mZ\big(
            x_1^n,x_2^n,\ldots,x_K^n,
            (\widehat{w}_1,\widehat{w}_2,\ldots,\widehat{w}_K)
            \mid 
            w_1,w_2,\ldots,w_K,y^n \big)$
    };

    \node (IT1) at (-5.5,0.9) {};
    \node[below=-0.25cm of IT1] {\small $w_1$};

    \node (IT2) at (-3.5,0.9) {};
    \node[below=-0.25cm of IT2] {\small $w_2$};

    \node (ITK) at (0.5,0.9) {};    
    \node[below=-0.25cm of ITK] {\small $w_K$};

    \node (IR) [
        rectangle,
        minimum width=0.6cm
    ] at (5,-0.9) {};
    \node[above=-0.3cm of IR] {\small $y^n$};

    \node (OT1) at ($(IT1.south)+(0,-1.55)$) {\small $x_1^n$};
    \node (OT2) at ($(IT2.south)+(0,-1.55)$) {\small $x_2^n$};
    \node (OTK) at ($(ITK.south)+(0,-1.55)$) {\small $x_K^n$};

    \node (OR) at ($(IR.north)+(-1,1.5)$)
        {\small $(\widehat w_1,\widehat w_2,\ldots,\widehat w_K)$};

    \node (N) [
        rectangle,
        draw,
        thick,
        minimum width=3cm,
        minimum height=2.2cm
    ] at (3,-3) {};

    \node [
        rectangle,
        minimum height=1.5cm,
        fill=gray!30
    ] at (3,-3) {
        \small
        $\mathsf{N}^{\otimes n}_{Y\mid X_1X_2\cdots X_K}$
    };

    \node (X1) at ($(N.west)+(0.15,-0.7)$) {};
    \node (X2) at ($(N.west)+(0.15,-0.25)$) {};
    \node (XD) at ($(N.west)+(0.15,0.2)$) {};
    \node (XK) at ($(N.west)+(0.15,0.7)$) {};

    \node (Y) at ($(N.east)+(-0.15,0)$) {};

    \node (Tx1) [above=1cm of IT1] {\sc Tx$1$};
    \node (Tx2) [above=1cm of IT2] {\sc Tx$2$};
    \node at (-1.5,2.3) {$\cdots$};
    \node (TxK) [above=1cm of ITK] {\sc Tx$K$};

    \node (Rx) [above=1cm of OR] {\sc Receiver};

    \node (C) [below=0cm of N] {\small\sc Multiple access channel};

    \draw[-{Latex[length=2.5mm]},thick]
        (Tx1) -- (IT1)
        node[pos=0.5,left] {$W_1$};

    \draw[-{Latex[length=2.5mm]},thick]
        (Tx2) -- (IT2)
        node[pos=0.5,left] {$W_2$};

    \draw[-{Latex[length=2.5mm]},thick]
        (TxK) -- (ITK)
        node[pos=0.5,left] {$W_K$};

    \draw[-{Latex[length=2.5mm]},thick]
        (OT1) |- (X1)
        node[pos=0.14,left] {$X_1^n$};

    \draw[-{Latex[length=2.5mm]},thick]
        (OT2) |- (X2)
        node[pos=0.175,left] {$X_2^n$};

    \node[above left= -0.4cm and 0.3cm of XD] {\small $\vdots$};

    \draw[-{Latex[length=2.5mm]},thick]
        (OTK) |- (XK)
        node[pos=0.30,left] {$X_K^n$};

    \draw[-{Latex[length=2.5mm]},thick]
        (Y) -| (IR)
        node[pos=0.68,right] {$Y^n$};

    \draw[-{Latex[length=2.5mm]},thick]
        (OR) -- (Rx)
        node[pos=0.5,right]
        {$(\widehat W_1,\widehat W_2,\ldots,\widehat W_K)$};

\end{tikzpicture}
\caption{NS-assisted coding scheme $\mZ$ for the $K$-sender
multiple access channel. `Tx$k$' denotes \emph{Transmitter $k$} for each $k\in [K]$.}
\label{fig:MAC_K}
\end{figure}
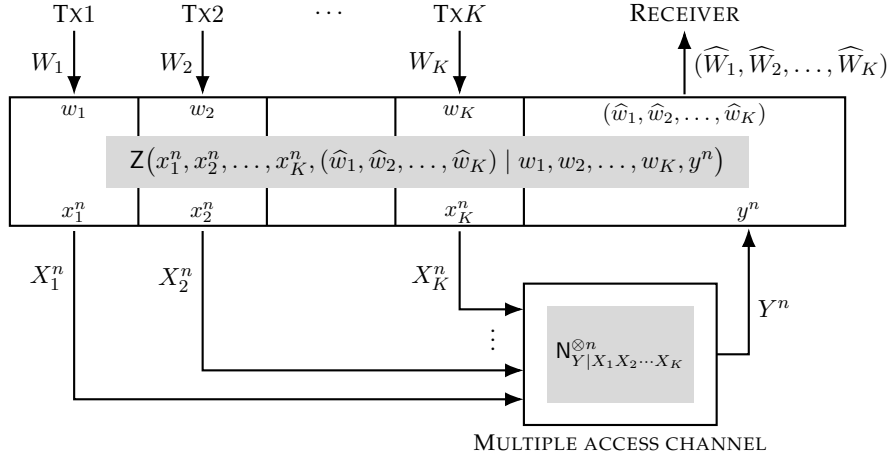

Let $M_1,M_2,\ldots, M_K, n$ be positive integers. 
An $(M_1,M_2,\ldots, M_K,n)$ NS-assisted coding scheme operating on $n$ channel uses is specified by a conditional pmf $\mZ\in \mathcal{P}(\mathcal{X}_1^n \times \mathcal{X}_2^n \times \cdots \times \mathcal{X}_K^n \times [M_1]\times [M_2] \times \cdots \times [M_K] \mid [M_1]\times [M_2] \times \cdots \times [M_K] \times \mathcal{Y}^n)$. 

Recall that a $(K+1)$-partite non-signaling correlation must not allow signaling (communication) across any bipartition. The $2^{K+1}-2$ NS conditions corresponding to all non-trivial bipartitions, are equivalent to the $K+1$ conditions that require only that each individual party cannot signal to the coalition of the remaining $K$ parties (without using the channel). Accordingly, we require that the conditional pmf $\mZ$ should satisfy the following $K+1$ conditions,
\begin{enumerate}[label=\textit{C\arabic*}:, start=0]
	\item $\sum_{\widehat{w}_1,\widehat{w}_2,\ldots, \widehat{w}_K}\mZ(x_1^n,x_2^n,\ldots,x_K^n,(\widehat{w}_1,\widehat{w}_2,\ldots, \widehat{w}_K)\mid w_1,w_2,\ldots, w_K, y^n)$ must be invariant under changes of $y^n$,
\end{enumerate}
which ensures no-signaling from the receiver to the coalition of the $K$ transmitters. Furthermore, we require
for each $k\in [K]$,
\begin{enumerate}[label=\textit{Ck}:]
	\item $\sum_{x_k^n}\mZ(x_1^n,x_2^n,\ldots,x_K^n,(\widehat{w}_1,\widehat{w}_2,\ldots, \widehat{w}_K)\mid w_1,w_2,\ldots, w_K, y^n)$ must be invariant under changes of $w_k$.
\end{enumerate}
Observe that \textit{Ck} ensures no-signaling from Transmitter $k$ to the coalition comprising the remaining $K-1$ transmitters and the receiver.

Such a scheme $\mZ$ is illustrated in Fig. \ref{fig:MAC_K}, represented as a black box shared across all parties where each party has its own input into and output out of the box $\mZ$, and works as follows. 
Let $(W_1,W_2,\ldots, W_K)$ be $K$ independent messages such that for each $k\in [K]$, $W_k$ is uniformly distributed over $[M_k]$. For each $k\in [K]$, Transmitter $k$ inputs $W_k$ to the box. The scheme generates $(X_1^n,X_2^n,\ldots, X_K^n)$, such that $X_k^n$ is provided to Transmitter $k$, according to the conditional pmf,
\begin{align}
&\mZ(x_1^n,x_2^n,\ldots,x_K^n\mid w_1,w_2,\ldots, w_K)\notag\\
&=\sum_{\widehat{w}_1, \widehat{w}_2,\ldots, \widehat{w}_K}\mZ(x_1^n,x_2^n,\ldots,x_K^n,(\widehat{w}_1,\widehat{w}_2,\ldots, \widehat{w}_K)\mid w_1,w_2,\ldots, w_K, y^n).
\end{align}
Note that this distribution does not depend on $y^n$, as required by \textit{C0}. 
The $(X_1^n,X_2^n,\ldots, X_K^n)$ generated by the scheme are then sent through $n$ uses of the multiple access channel.
The receiver receives $Y^n$ from the channel and inputs $Y^n$ to the box $\mZ$. 
The scheme provides $(\widehat{W}_1,\widehat{W}_2,\ldots, \widehat{W}_K)$ to the receiver, according to the conditional joint pmf $\mZ(\widehat{w}_1, \widehat{w}_2, \ldots, \widehat{w}_K \mid w_1,w_2,\ldots, w_K, x_1^n, x_2^n,\ldots, x_K^n, y^n)$.

The joint distribution for the random variables $(W_1,\ldots, W_K, X_1^n,  \ldots, X_K^n, Y^n,  \widehat{W}_1,  \ldots, \widehat{W}_K)$ is,
\begin{align} \label{eq:joint_dist}
	&\Pr(W_k=w_k, X_k^n=x_k^n, Y^n=y^n, \widehat{W}_k = \widehat{w}_k, \forall k\in [K]) \notag \\
	&=\big( \prod_{k=1}^K M_k \big)^{-1}\prod_{i=1}^n \mN_{Y\mid X_1X_2\cdots X_K}(y_i\mid x_{1,i},x_{2,i}, \ldots, x_{K,i}) \notag \\
	&\hspace{1cm} \times  \mZ(x_1^n,x_2^n,\ldots,x_K^n,(\widehat{w}_1,\widehat{w}_2,\ldots, \widehat{w}_K)\mid w_1,w_2,\ldots, w_K, y^n).
\end{align}
For each $k\in [K]$, define the probability of decoding error for the message $W_k$ as,
\begin{align}
P_{e,k}(\mZ) = \Pr(\widehat{W}_k\neq W_k).
\end{align}

\subsection{Achievable rate tuples and capacity region}
Given the $K$-sender DM-MAC $\mN_{Y\mid X_1X_2\cdots X_K}$, a rate tuple $(R_1,R_2,\ldots, R_K)\in \mathbb{R}_{\geq 0}^K$ is said to be achievable by NS-assisted coding schemes if there exists a sequence of $(M_{1}^{(n)}, \ldots, M_{K}^{(n)},n)$ NS-assisted coding schemes $(\mZ^{(n)})_{n\in \mathbb{N}}$ such that 
\begin{align}
\lim_{n\to \infty} P_{e,k} (\mZ^{(n)}) &= 0,\\
\mbox{and }\liminf_{n\to \infty} \frac{\log_2 M_{k}^{(n)}}{n} &\geq R_k, \quad\forall k\in [K].
\end{align}
The NS-assisted capacity region $\mathcal{C}^{\NS}(\mN_{Y\mid X_1X_2\cdots X_K})$ is defined as the closure of all rate tuples achievable by NS-assisted coding schemes.

In contrast, we use $\mathcal{C}(\mN_{Y\mid X_1X_2\cdots X_K})$ to denote the classical capacity region of the $K$-sender DM-MAC \cite[Theorem 4.5]{NIT}. 

\subsection{Sum-rate capacity}
Given the $K$-sender DM-MAC $\mN_{Y\mid X_1X_2\cdots X_K}$, let us define its NS-assisted sum-rate capacity as
\begin{align}
	C_{\Sigma}^{\NS}(\mN_{Y\mid X_1X_2\cdots X_K}) = \max_{\substack{(R_1,R_2,\ldots, R_K) \in\\ \mathcal{C}^{\NS}(\mN_{Y\mid X_1X_2\cdots X_K})}} (R_1+R_2+\cdots+R_K),
\end{align}
and its classical sum-rate capacity as
\begin{align}
	C_{\Sigma}(\mN_{Y\mid X_1X_2\cdots X_K}) = \max_{\substack{(R_1,R_2,\ldots, R_K) \in\\ \mathcal{C}(\mN_{Y\mid X_1X_2\cdots X_K})}} (R_1+R_2+\cdots+R_K).
\end{align}

\section{Results}
\subsection{NS-assisted capacity region of $K$-sender MAC}\label{sec:Kcapregion}
The main result of this paper is stated in the following theorem.
\begin{theorem}[NS-assisted $K$-sender MAC capacity region] \label{thm:NSMAC}
	For a $K$-sender discrete memoryless multiple access channel $\mN_{Y\mid X_1X_2\cdots X_K}$, the NS-assisted capacity region $\mathcal{C}^{\NS}(\mN_{Y\mid X_1X_2\cdots X_K})$ is the closure of the convex hull of all rate tuples $(R_1,R_2, \ldots, R_K) \in \mathbb{R}_{\geq 0}^K$ satisfying
	\begin{align}
		\sum_{k\in \mathcal{K}} R_k \leq I(X_{\mathcal{K}}; Y \mid X_{[K] \setminus \mathcal{K}}),~~ \forall \emptyset \neq \mathcal{K} \subseteq [K]
	\end{align}
	for some joint pmf $\mP_{X_1X_2\cdots X_K}\in \mathcal{P}(\mathcal{X}_1\times \mathcal{X}_2\times \cdots \times \mathcal{X}_K)$, where all (conditional) mutual informations are defined according to $\mP_{X_1X_2\cdots X_K}\times \mN_{Y\mid X_1X_2\cdots X_K}$. Here, $X_{\mathcal{K}}$ denotes $(X_k)_{k\in \mathcal{K}}$, and similarly $X_{[K]\setminus \mathcal{K}} = (X_k)_{k\in [K] \setminus \mathcal{K}}$.
	
	The region can be equivalently expressed by introducing an auxiliary time-sharing variable $Q$, as
	\begin{align}
		&\mathcal{C}^{\NS}(\mN_{Y\mid X_1X_2\cdots X_K}) \notag \\
		&= \bigcup_{\mP_{Q X_1X_2\cdots X_K}}
		\left\{
		\begin{array}{l}
			(R_1,R_2,\ldots, R_K) \in \mathbb{R}_{\geq 0}^K\colon\\ \sum_{k\in \mathcal{K}} R_k \leq I(X_{\mathcal{K}}; Y \mid X_{[K] \setminus \mathcal{K}},Q),~~ \forall \emptyset \neq \mathcal{K} \subseteq [K]
		\end{array}
		\right.
	\end{align}
	where the union is over all joint pmfs $\mP_{QX_1X_2\cdots X_K} \in \mathcal{P}(\mathcal{Q} \times \mathcal{X}_1 \times \mathcal{X}_2 \times \cdots \times \mathcal{X}_K)$ and it suffices to have $|\mathcal{Q}|\leq K$. All (conditional) mutual informations are defined according to $\mP_{Q X_1X_2 \cdots X_K} \mN_{Y\mid X_1X_2\cdots X_K}$.
\end{theorem}
\noindent The proof of converse is presented in Appendix \ref{proof:converse}. The converse argument essentially follows from the two-sender converse (outer bound) in \cite[Corollary 42]{fawzi2024MAC}.

The proof of achievability is presented in Appendix  \ref{proof:achievability}. The special case for the two-sender setting is presented in Section \ref{proof:two_sender}. The proof first constructs an unnormalized sub-distribution which satisfies the NS conditions \textit{C1} and \textit{C2}. A normalization process then ensures that it satisfies \textit{C0} and is a valid pmf for sufficiently large $n$. Bearing some similarity to NS-assisted broadcast channels \cite{Yao_Jafar_NSSato}, the outputs $(X_1^n,X_2^n)$ of the scheme on the transmitters' side are generated according to a distribution $\mQ \in \mathcal{P}(\mathcal{X}_1\times \mathcal{X}_2)$ which does not depend on the messages. Therefore $X_1^n,X_2^n$ may be viewed as some random keys generated by the scheme, that are used to authenticate the receiver's input $Y^n$ later. A remarkable feature of the scheme is that the construction of $\mQ$ turns out to be a convex combination $\mQ = (1-\frac{1}{n})\mP_{X_1X_2}^{\otimes n} + \frac{1}{n}\mP_{X_1}^{\otimes n} \mP_{X_2}^{\otimes n}$. This is different from the distribution used by Fawzi and Ferm\'{e} in \cite{fawzi2024MAC}, where $(X_1^n,X_2^n)$ was generated according to $\mP_{X_1X_2}^{\otimes n}$ conditioned on $(X_1^n, X_2^n)$ being jointly typical. The factor $(1-\frac{1}{n})$ eventually translates to a value multiplied to the probability of successful decoding, and thus does not hurt the asymptotic performance, whereas the factor $\frac{1}{n}$ is sufficient to dominate any function that decays exponentially fast in $n$, so that it can satisfy all conditions required by an NS-assisted coding scheme.

\subsection{Resolving an open question in \cite{fawzi2024MAC}}
Reference \cite{fawzi2024MAC} studied the two-sender DM-MACs $\mN_{Y\mid X_1X_2}$ with NS assistance, and it was proved that the region in Theorem \ref{thm:NSMAC} (for $K=2$) matches $\mathcal{C}^{\overline{\NS}}(\mN_{Y\mid X_1X_2})$, where $\mathcal{C}^{\overline{\NS}}(\mN_{Y\mid X_1X_2})$ is defined as the capacity region for a set of relaxed non-signaling (relaxed-NS) coding schemes over channel $\mN_{Y\mid X_1X_2}$ \cite[Theorem 34]{fawzi2024MAC}. The definition of the relaxed-NS coding schemes is obtained by relaxing certain non-signaling constraints in the definition of non-signaling coding schemes. Consequently, the class of relaxed-NS coding schemes is a superset of NS-assisted coding schemes. It was left as an open problem in \cite{fawzi2024MAC} whether $\mathcal{C}^{\NS}(\mN_{Y\mid X_1X_2}) = \mathcal{C}^{\overline{\NS}}(\mN_{Y\mid X_1X_2})$, even for the binary-adder MAC, i.e., $Y=X_1+X_2$, where $X_k\in \{0,1\}, \forall k\in \{1,2\}$ and the addition is over integers. 
Theorem \ref{thm:NSMAC} answers this open question in the affirmative, showing that the two capacities coincide. As a direct corollary, for the binary-adder channel, the NS-assisted sum-rate capacity $C^{\NS}_{\Sigma}$ is thus shown to be equal to $\log_2(3)\approx 1.585$.

\subsection{Extremal improvement over classical capacity}\label{sec:extremal}
For any DM-MAC $\mN_{Y\mid X_1X_2\cdots X_K}$, Theorem \ref{thm:NSMAC} allows us to compare its NS-assisted capacity region $\mathcal{C}^{\NS}(\mN_{Y\mid X_1 X_2 \cdots X_K})$ with its classical capacity region $\mathcal{C}(\mN_{Y\mid X_1 X_2 \cdots X_K})$ \cite[Theorem 4.5]{NIT}. The difference between the two regions is evident, in that the allowed input distribution $\mP_{X_1X_2\cdots X_K}$ in $\mathcal{C}(\mN_{Y\mid X_1 X_2 \cdots X_K})$ must be independent across the senders, whereas in $\mathcal{C}^{\NS}(\mN_{Y\mid X_1 X_2 \cdots X_K})$ this can be chosen to be arbitrarily dependent across the senders. 

It is then natural to ask how much the two regions can differ from each other, and in particular, how large the extremal ratio of the sum-rate capacities $C^{\NS}_{\Sigma}/C_{\Sigma}$ can be. Reference \cite[Theorem 3]{Yao_Jafar_Unlock} proves that this ratio cannot be more than $K$ for any $K$-sender MAC.\footnote{This result seemingly contradicts the findings reported in \cite[Table II]{QMAC_Yun} by Yun et al. We clarify this issue in Appendix \ref{sec:yunetal} by pointing out a minor oversight in the calculation of the non-signaling assisted sum-rate capacity in \cite[Table II]{QMAC_Yun}.}  

Here we answer this question by showing that there exists a $K$-sender DM-MAC for which the ratio $C^{\NS}_{\Sigma}/C_{\Sigma}$ can be arbitrarily close to $K$. The channel is referred to as the \emph{sum-product-lock channel}, because the channel is \emph{locked} (output is erased) unless a certain sum of inputs matches a certain product of inputs. The definition is as follows.
\begin{definition}[Sum-product-lock channel] \label{def:sum_prod_lock_channel}
	Let $K\geq 2$ and let $q$ be a power of a prime. Let $\mathbb{F}_q$ denote the finite field of order $q$ and $\mathbb{F}_q^{\times} \triangleq \mathbb{F}_q \setminus \{0\}$.  The sum-product-lock channel has input alphabets $\mathcal{X}_k = \mathbb{F}_q^{\times} \times \mathbb{F}_q, \forall k\in [K]$, and output alphabet $\mathcal{Y} = (\mathbb{F}_q^{\times})^K \cup \{\phi\}$, where $\phi$ denotes an erasure symbol. Given that $X_k = (a_k,b_k)$ where $a_k\in \mathbb{F}_q^{\times}$ and $b_k\in \mathbb{F}_q, \forall k\in [K]$, the channel's law is that 
	\begin{align} \label{eq:def_channel_sum_prod_lock}
	Y = \begin{cases}
		(a_1,a_2,\ldots, a_K), & \mbox{if}~ \sum_{j=1}^K b_j = \prod_{i=1}^K a_i\\
		\phi, & \mbox{otherwise}
	\end{cases}.
\end{align}
\end{definition}
Let us note that the sum-product-lock channel in Definition \ref{def:sum_prod_lock_channel} is deterministic, and it belongs to a class of MACs that have been studied in the NS/entanglement-assisted literature \cite{Quek_Shor, leditzky2020playing, seshadri2023separation, QMAC_Yun}, in which a non-local game is transformed into a MAC. This idea is due to Quek and Shor \cite{Quek_Shor} and was later developed by Leditzky et al. \cite{leditzky2020playing}. Specifically, in a $K$-player non-local game, for each $k\in [K]$, Player $k$ answers $b_k$ based on its local input $a_k$. They win the game if and only if $b_1,b_2,\ldots,b_K$ and $a_1,a_2,\cdots, a_K$ satisfy a winning condition (the condition in \eqref{eq:def_channel_sum_prod_lock} as an example). In the corresponding MAC, both the questions $(a_1,a_2,\ldots,a_K)$ and the answers $(b_1,b_2,\ldots,b_K)$ are channel inputs chosen by the transmitters, and the channel decides whether or not to reveal (with certain probability) the questions to the receiver depending on whether the inputs satisfy that winning condition.

We have the following theorem.
\begin{theorem}\label{thm:sum_prod_lock}
	For the sum-product-lock channel defined in \eqref{eq:def_channel_sum_prod_lock}, the NS-assisted sum-rate capacity $C_{\Sigma}^{\NS} \geq K \log_2 (q-1)$, whereas its classical sum-rate capacity $C_{\Sigma} \leq \log_2 q + o_q(\log_2(q))$.
\end{theorem}
\noindent The proof of $C^{\NS}_{\Sigma} \geq K\log_2(q-1)$ is presented in Appendix \ref{proof:sum_prod_lock_NSa}. The proof of $C_{\Sigma} \leq \log_2 q + o_q(\log_2 q)$ is presented in Appendix \ref{proof:sum_prod_lock_Cc}. For intuition, we provide the following proof sketch.
\begin{proof}[(Sketch)]
The proof that $C_{\Sigma}^{\NS} \geq K\log_2 (q-1)$ (in Appendix \ref{proof:sum_prod_lock_NSa}) is straightforward using Theorem \ref{thm:NSMAC}. In fact, an alternative proof (also presented in Appendix \ref{proof:sum_prod_lock_NSa}) shows that the sum-rate $K\log_2(q-1)$ is achievable with only transmitters' side NS assistance.
The challenging part is the proof of $C_{\Sigma} \leq \log_2 q +o_q(\log_2 q)$ (in Appendix \ref{proof:sum_prod_lock_Cc}). This is done by proving that, for any product distribution $\mP_{X_1}\mP_{X_2}\cdots \mP_{X_K}$, the output entropy $H(Y)$ is at most $\log_2 q+ o_q(\log_2 q)$, thus proving the bound on $C_{\Sigma}$. 
Specifically, we  associate with the channel a $K$-partite $K$-uniform hypergraph, where the $k^{th}$ vertex class consists of the possible inputs of Transmitter $k$, and a $K$-tuple of vertices forms a hyperedge if and only if the corresponding channel inputs result in a non-erasure output. 
Next, given any product input distribution, we partition each transmitter's input alphabet into probability-level classes. Inputs whose individual probabilities fall below a prescribed threshold are grouped into a residual class, whose total contribution to the output entropy is negligible. 
Now condition on any fixed choice of non-residual probability-level classes. This restricts each vertex class to its selected probability-level class, thereby inducing a subhypergraph. Define its edge density as the fraction of input tuples in the corresponding Cartesian product that result in non-erasure. 
If the conditional non-erasure probability is negligible, then the conditional entropy of the output is negligible as well. Otherwise, the induced subhypergraph has sufficiently large edge density, and it can be shown that all but at most one of the selected input classes have cardinality $q^{o(1)}$, due to a conditionally $\mathcal{K}_{2,2}$-free property (see Definition \ref{def:cond_K22free}). Consequently, conditioned on these probability-level classes, all but at most one of the transmitters have input entropy $o_q(\log_2 q)$. The remaining transmitter can only signal through one $\mathbb{F}_q^{\times}$-valued coordinate (in addition to the erasure symbol), and therefore contributes at most $\log_2 q+O_q(1)$ bits. Taking the average, the conditional entropy of $Y$ given the class indices is bounded by $\log_2 q+o_q(\log_2 q)$. Since the probability partition process ensures that the entropy of the class indices is $o_q(\log_2 q)$, we have an upper bound $H(Y)\leq \log_2 q+o_q(\log_2 q)$.
\end{proof}

\section{Achievability of Theorem \ref{thm:NSMAC}: Two-sender case} \label{proof:two_sender}
It suffices to prove that, for each $\mP_{X_1X_2}\in \mathcal{P}(\mathcal{X}_1\times \mathcal{X}_2)$, $(R_1,R_2)\in \mathbb{R}^2_{\geq 0}$ is achievable by NS-assisted coding schemes if
\begin{align}
	R_1 &< I(X_1;Y\mid X_2)\\
	R_2 &< I(X_2;Y\mid X_1) \\
	R_1+R_2 &< I(X_1,X_2;Y)
\end{align}
whenever $\min\{I(X_1;Y\mid X_2), I(X_2;Y\mid X_1)\} >0$. This avoids degenerate cases with fewer senders.
Together with time-sharing, this proves the achievability of Theorem \ref{thm:NSMAC} for the two-sender case. Henceforth, $\mP_{X_1X_2}$ is fixed to a pmf in $\mathcal{P}(\mathcal{X}_1\times \mathcal{X}_2)$.

\subsection{NS conditions for $K=2$}
For $K=2$, an $(M_1,M_2,n)$ NS-assisted coding scheme over $n$ uses of the MAC is specified by a conditional pmf $\mZ(x_1^n,x_2^n,(\widehat{w}_1,\widehat{w}_2)\mid w_1,w_2,y^n)$ for which the following set of conditions hold.
\begin{enumerate}[label=\textit{C\arabic*}:, start=0]
	\item $\sum_{\widehat{w}_1,\widehat{w}_2}\mZ(x_1^n,x_2^n,(\widehat{w}_1,\widehat{w}_2)\mid w_1,w_2, y^n)$ is invariant under changes of $y^n \in \mathcal{Y}^n$;
	\item $\sum_{x_1^n}\mZ(x_1^n,x_2^n,(\widehat{w}_1,\widehat{w}_2)\mid w_1,w_2,y^n)$ must be invariant under changes of $w_1$;
	\item $\sum_{x_2^n}\mZ(x_1^n,x_2^n,(\widehat{w}_1,\widehat{w}_2)\mid w_1,w_2,y^n)$ must be invariant under changes of $w_2$.
\end{enumerate}

\subsection{Symmetrization}
As a preliminary step, we adopt part of the symmetrization process in \cite{fawzi2024MAC} to reduce the number of parameters in $\mZ$ without affecting optimality. The idea is also referred to as  twirling in \cite{cubitt2011zero}. 

For $m\in \mathbb{N}$, let $\mathcal{S}_m$ denote the group of permutations from $[m] \to [m]$. It follows that $|\mathcal{S}_m| = m!$. 
Given any $(M_1,M_2,n)$ NS-assisted scheme $\mZ_o(x_1^n,x_2^n,(\widehat{w}_1,\widehat{w}_2)\mid w_1,w_2,y^n)$, define
\begin{align}
	&\mZ(x_1^n,x_2^n,(\widehat{w}_1,\widehat{w}_2) \mid w_1,w_2,y^n)\notag\\
	&\triangleq \frac{1}{M_1! M_2!}\sum_{\substack{\pi_1\in \mathcal{S}_{M_1}\\ \pi_2 \in \mathcal{S}_{M_2}}}\mZ_o(x_1^n,x_2^n,(\pi_1(\widehat{w}_1),\pi_2(\widehat{w}_2))\mid \pi_1(w_1),\pi_2(w_2),y^n)
\end{align}
In other words, $\mZ$ is a convex combination of (in total $M_1! M_2!)$ NS-assisted coding schemes, each of which is obtained from $\mZ_o$ by permuting the labels of $W_1, \widehat{W}_1$ by $\pi_1$, and the labels of $W_2,\widehat{W}_2$ by $\pi_2$. It can be readily verified that $\mZ$ is also an NS-assisted coding scheme, and that $P_{e,k}(\mZ) = P_{e,k}(\mZ_o), \forall k\in \{1,2\}$. Let us call $\mZ$ the twirled scheme, which is symmetric in that it adopts the form
\begin{align} \label{eq:twirl_1}
	&\mZ(x_1^n,x_2^n,(\widehat{w}_1,\widehat{w}_2) \mid w_1,w_2,y^n) \notag \\
	&=
	\begin{cases}
		z_{\cmark\!, \cmark}(x_1^n,x_2^n,y^n), & \mbox{if}~ \widehat{w}_1=w_1, \widehat{w}_2=w_2\\
		z_{\cmark\!, \xmark}(x_1^n,x_2^n,y^n), & \mbox{if}~ \widehat{w}_1=w_1, \widehat{w}_2\neq w_2\\ 
		z_{\xmark\!, \cmark}(x_1^n,x_2^n,y^n), & \mbox{if}~ \widehat{w}_1\neq w_1, \widehat{w}_2=w_2\\
		z_{\xmark\!, \xmark}(x_1^n,x_2^n,y^n), & \mbox{if}~ \widehat{w}_1\neq w_1, \widehat{w}_2\neq w_2
	\end{cases}
\end{align}
i.e., the value only depends on the indicators $\mathbb{I}(\widehat{w}_1=w_1)$, $\mathbb{I}(\widehat{w}_2=w_2)$, together with $(x_1^n,x_2^n,y^n)$. 
It follows that the marginal distribution 
\begin{align}
	&\mZ(x_1^n,x_2^n\mid w_1,w_2,y^n)\notag \\
	&= z_{\cmark\!, \cmark}(x_1^n,x_2^n,y^n)  \notag \\
	&~~~~+ (M_2-1)z_{\cmark\!, \xmark}(x_1^n,x_2^n,y^n) \notag \\
	&~~~~+ (M_1-1)z_{\xmark\!, \cmark}(x_1^n,x_2^n,y^n) \notag \\
	&~~~~+ (M_1-1)(M_2-1)z_{\xmark\!, \xmark}(x_1^n,x_2^n,y^n)
\end{align}
which does not depend on $(w_1,w_2)$, nor $y^n$ because of non-signaling.

Note that for $\mZ$, the probability of successfully decoding both messages is
\begin{align}
	&\Pr(\widehat{W}_1=W_1, \widehat{W}_2=W_2) \notag \\
	&\stackrel{\eqref{eq:joint_dist}}{=} \frac{1}{M_1M_2}\sum_{\substack{w_1,w_2\\ x_1^n,x_2^n,y^n}} \mZ(x_1^n,x_2^n,(w_1,w_2)\mid w_1,w_2, y^n) \times \mN^{\otimes n}_{Y\mid X_1X_2}(y^n\mid x_1^n,x_2^n) \\
	&= \sum_{x_1^n,x_2^n,y^n}z_{\cmark\!, \cmark}(x_1^n,x_2^n,y^n) \mN_{Y\mid X_1X_2}^{\otimes n}(y^n\mid x_1^n,x_2^n).
\end{align}

The remaining part of the proof is to construct reliable NS-assisted coding schemes with a special structure motivated from the twirled scheme.

\subsection{A useful sub-distribution}
Let $\epsilon \in (0,1)$. 
In the following, we write $\mathcal{T}^{(n)}_\epsilon(\mP_{X_1X_2}\mN_{Y\mid X_1X_2})$ simply as $\mathcal{T}^{(n)}_\epsilon$.
For each $(x_1^n, x_2^n, y^n)\in \mathcal{X}_1^n \times \mathcal{X}_2^n \times \mathcal{Y}^n$, let
\begin{align} \label{eq:def_a12}
	a_{12}(x_1^n,x_2^n,y^n) \triangleq \lambda_n \times \mP_{X_1X_2}^{\otimes n}(x_1^n,x_2^n)\times \mathbb{I}\big( (x_1^n,x_2^n,y^n) \in \mathcal{T}^{(n)}_\epsilon \big), 
\end{align}
where $\lambda_n \in [0,1]$, 
and
\begin{align} 
	& a_{1}(x_1^n,y^n) \triangleq \sum_{x_2^n} a_{12} (x_1^n,x_2^n,y^n), \label{eq:def_a1} \\
	& a_{2}(x_2^n,y^n) \triangleq \sum_{x_1^n} a_{12}(x_1^n,x_2^n,y^n), \label{eq:def_a2} \\
	& a(y^n) \triangleq \sum_{x_1^n,x_2^n} a_{12} (x_1^n,x_2^n,y^n). \label{eq:def_a}
\end{align}
The role of $\lambda_n$ will become clear later. It represents a penalty in the probability of successful decoding. Reliable decoding will only require that $\lambda_n \to 1$ as $n\to \infty$. Consequently, we have some freedom in choosing $\lambda_n$. It turns out that $\lambda_n = 1-\frac{1}{n}$ suffices for our purpose, and the key reason is that $\frac{1}{n}$ approaches $0$ much slower than $2^{-\gamma n}$ for any $\gamma >0$.

To construct a scheme $\mZ$ that satisfies the NS conditions \textit{C0}--\textit{C2}, we start by defining an (unnormalized) sub-distribution 
\begin{align} \label{eq:def_Ztilde}
	&\widetilde{\mZ}(x_1^n,x_2^n,(\widehat{w}_1,\widehat{w}_2)\mid w_1,w_2,y^n) \notag \\
	&\triangleq \begin{cases}
		a_{12}(x_1^n,x_2^n,y^n), & \mbox{if}~ \widehat{w}_1=w_1, \widehat{w}_2 = w_2\\
		a_{1}(x_1^n,y^n) \mP_{X_2}^{\otimes n}(x_2^n), & \mbox{if}~ \widehat{w}_1=w_1, \widehat{w}_2 \neq w_2\\
		a_{2}(x_2^n,y^n) \mP_{X_1}^{\otimes n}(x_1^n) , & \mbox{if}~ \widehat{w}_1\neq w_1, \widehat{w}_2 = w_2\\
		a(y^n) \mP_{X_1}^{\otimes n}(x_1^n) \mP_{X_2}^{\otimes n}(x_2^n), & \mbox{if}~ \widehat{w}_1\neq w_1, \widehat{w}_2 \neq w_2
	\end{cases}
\end{align}
Note that its value only depends on the indicators $\mathbb{I}(\widehat{w}_1=w_1)$, $\mathbb{I}(\widehat{w}_2=w_2)$, together with $(x_1^n,x_2^n,y^n)$. This is motivated by \eqref{eq:twirl_1}.

The value $\widetilde{\mZ}(x_1^n,x_2^n,(\widehat{w}_1,\widehat{w}_2)\mid w_1,w_2,y^n)$ is always non-negative. Moreover, the marginal sub-distribution
\begin{align}
	&\widetilde{\mZ}(x_2^n,(\widehat{w}_1,\widehat{w}_2)\mid w_1,w_2,y^n) \notag \\
	&\triangleq \sum_{x_1^n} \widetilde{\mZ} (x_1^n,x_2^n,(\widehat{w}_1,\widehat{w}_2)\mid w_1,w_2,y^n) \\
	&= \begin{cases}
		a_{2}(x_2^n,y^n), & \mbox{if}~ \widehat{w}_2 = w_2\\
		a(y^n) \mP_{X_2}^{\otimes n}(x_2^n), & \mbox{if}~ \widehat{w}_2 \neq w_2
	\end{cases}
\end{align}
and thus is invariant under changes of $w_1$. 
Similarly, 
\begin{align}
	& \widetilde{\mZ}(x_1^n,(\widehat{w}_1,\widehat{w}_2)\mid w_1,w_2,y^n) \notag \\
	&\triangleq\sum_{x_2^n} \widetilde{\mZ} (x_1^n,x_2^n,(\widehat{w}_1,\widehat{w}_2)\mid w_1,w_2,y^n) \\
	&= \begin{cases}
		a_{1}(x_1^n,y^n), & \mbox{if}~ \widehat{w}_1 = w_1\\
		a(y^n) \mP_{X_1}^{\otimes n}(x_1^n), & \mbox{if}~ \widehat{w}_1 \neq w_1
	\end{cases}
\end{align}
and thus is invariant under changes of $w_2$. Therefore, the sub-distribution $\widetilde{\mZ}$ satisfies \textit{C1} and \textit{C2}. However, $\widetilde{\mZ}$ may not satisfy \textit{C0}. It may also violate the normalization condition that the sum over $x_1^n,x_2^n,\widehat{w}_1,\widehat{w}_2$ should be equal to $1$ for every $(w_1,w_2,y^n)$. 

\subsection{Normalizing the distribution}
To tackle these problems, let
\begin{align} \label{eq:def_Z}
	&\mZ(x_1^n,x_2^n,(\widehat{w}_1,\widehat{w}_2)\mid w_1,w_2,y^n) \notag \\
	&\triangleq \widetilde{\mZ}(x_1^n,x_2^n,(\widehat{w}_1,\widehat{w}_2)\mid w_1,w_2,y^n) \notag \\
	&~~~~+ \frac{\mQ(x_1^n,x_2^n)- \widetilde{\mZ}(x_1^n,x_2^n\mid w_1,w_2,y^n)}{M_1M_2} 
\end{align}
where $\mQ\in \mathcal{P}(\mathcal{X}_1^n \times \mathcal{X}_2^n)$ is a valid pmf to be chosen later, and $\widetilde{\mZ}(x_1^n,x_2^n\mid w_1,w_2,y^n)$ is the marginal (unnormalized) distribution obtained by  summing over all $(\widehat{w}_1,\widehat{w}_2)$ in $\widetilde{\mZ}(x_1^n,x_2^n,(\widehat{w}_1,\widehat{w}_2)\mid w_1,w_2,y^n)$.
The form of $\mZ$ in \eqref{eq:def_Z} then satisfies \textit{C0}. To see this, note that 
\begin{align}
	&\sum_{\widehat{w}_1, \widehat{w}_2} \mZ(x_1^n,x_2^n,(\widehat{w}_1,\widehat{w}_2)\mid w_1,w_2,y^n) \notag \\
	&= \widetilde{\mZ}(x_1^n,x_2^n \mid w_1,w_2,y^n) + \mQ(x_1^n,x_2^n) - \widetilde{\mZ}(x_1^n,x_2^n \mid w_1,w_2,y^n) \\
	&= \mQ(x_1^n,x_2^n) \label{eq:Qoperational}
\end{align}
which is invariant under changes of $y^n$ (and also $w_1,w_2$). It can be readily verified that $\mZ$ also satisfies \textit{C1} and \textit{C2}, based on the fact that $\widetilde{\mZ}$ satisfies \textit{C1} and \textit{C2}. Meanwhile, $\mZ$ satisfies the normalization condition, that $\sum_{x_1^n,x_2^n, \widehat{w}_1, \widehat{w}_2} \mZ(x_1^n,x_2^n,(\widehat{w}_1,\widehat{w}_2)\mid w_1,w_2,y^n) = \sum_{x_1^n,x_2^n} \mQ(x_1^n,x_2^n) = 1$. 

Let us point out that \eqref{eq:Qoperational} shows that $\mQ(x_1^n,x_2^n)$ is the distribution that the scheme uses to generate $(X_1^n,X_2^n)$. Since it is invariant under changes of $(w_1,w_2,y)$, $(X_1^n,X_2^n)$ can be thought of as a pair of (correlated) keys, generated independently of the messages, that are sent through the MAC. 
Moreover, provided that $\mZ$ is a valid scheme, for every $(x_1^n,x_2^n)$ on the support of $\mQ$, the chain rule implies that
\begin{align}
	& \mZ(\widehat{w}_1,\widehat{w}_2\mid w_1,w_2,x_1^n,x_2^n,y^n) \notag \\
	&= \frac{\mZ(x_1^n,x_2^n,(\widehat{w}_1,\widehat{w}_2)\mid w_1,w_2,y^n)}{\mQ(x_1^n,x_2^n)}\\
	&= \frac{\widetilde{\mZ}(x_1^n,x_2^n,(\widehat{w}_1,\widehat{w}_2)\mid w_1,w_2,y^n)}{\mQ(x_1^n,x_2^n)} \notag \\
	&~~~~+ \frac{1-\widetilde{\mZ}(x_1^n,x_2^n\mid w_1,w_2,y^n)/\mQ(x_1^n,x_2^n)}{M_1M_2}
\end{align}
This is the probability that the scheme outputs $\widehat{W}_1= \widehat{w}_1, \widehat{W}_2=\widehat{w}_2$, conditioned on the receiver's input $y^n$, and the keys being $(x_1^n,x_2^n)$. Note that the value only depends on the indicators $\mathbb{I}(\widehat{w}_1=w_1)$ and $\mathbb{I}(\widehat{w}_2=w_2)$, together with $(x_1^n,x_2^n,y^n)$. 

\subsection{Validity of the solution}
Note that $\mZ$ satisfies \textit{C0}--\textit{C2} and the normalization condition. For $\mZ$ to be a valid pmf (and thus a valid NS scheme), it remains to satisfy the non-negativity condition, that $\mZ(x_1^n,x_2^n,(\widehat{w}_1,\widehat{w}_2)\mid w_1,w_2,y^n)\geq 0$. According to \eqref{eq:def_Z}, it suffices to have 
\begin{align} \label{eq:cond_ng}
	 \widetilde{\mZ}(x_1^n,x_2^n\mid w_1,w_2,y^n) \leq \mQ(x_1^n,x_2^n) 
\end{align}
for all $x_1^n,x_2^n, w_1,w_2,y^n$.

To this end, let us first observe that
\begin{align}
& \widetilde{\mZ}(x_1^n,x_2^n\mid w_1,w_2,y^n) \notag \\
	&\triangleq \sum_{\widehat{w}_1, \widehat{w}_2} \widetilde{\mZ}(x_1^n,x_2^n,(\widehat{w}_1,\widehat{w}_2) \mid w_1,w_2,y^n)\notag\\
	&\stackrel{\eqref{eq:def_Ztilde}}{=} a_{12}(x_1^n,x_2^n,y^n) \notag \\
	&~~~~+ (M_2-1) a_1(x_1^n,y^n) \mP_{X_2}^{\otimes n}(x_2^n) \notag \\
	&~~~~+ (M_1-1) a_2(x_2^n,y^n) \mP_{X_1}^{\otimes n}(x_1^n) \notag \\
	&~~~~+ (M_1-1)(M_2-1) a(y^n) \mP_{X_1}^{\otimes n}(x_1^n) \mP_{X_2}^{\otimes n}(x_2^n).  \label{eq:expansion}
\end{align} 
Then, let us bound $a_1(x_1^n,y^n)$ as follows.
\begin{align}
	&\hspace{-10pt} a_1(x_1^n,y^n) \notag\\
	&\hspace{-10pt} \stackrel{\eqref{eq:def_a12},\eqref{eq:def_a1}}{=} \sum_{x_2^n} \lambda_n \mP_{X_1X_2}^{\otimes n}(x_1^n,x_2^n)\times \mathbb{I}\big( (x_1^n,x_2^n,y^n) \in \mathcal{T}^{(n)}_\epsilon \big) \\
	&\leq  \sum_{x_2^n} \mP_{X_1X_2}^{\otimes n}(x_1^n,x_2^n)\times \mathbb{I}\big( (x_1^n,x_2^n,y^n) \in \mathcal{T}^{(n)}_\epsilon \big) \label{eq:use_lambda_bound} \\
	&=\mP_{X_1}^{\otimes n}(x_1^n) \sum_{x_2^n} \mP_{X_2\mid X_1}^{\otimes n}(x_2^n \mid x_1^n)\times \mathbb{I}\big( (x_1^n,x_2^n,y^n) \in \mathcal{T}^{(n)}_\epsilon \big)\\
	&=\mP_{X_1}^{\otimes n}(x_1^n) \times \Pr\big( (x_1^n,X_2^n,y^n) \in   \mathcal{T}^{(n)}_\epsilon \big) \label{eq:def_X2} \\
	&\leq \mP_{X_1}^{\otimes n}(x_1^n) \times 2^{-n(I(X_2;Y\mid X_1)-\delta_1(\epsilon))} \label{eq:a1_bound}
\end{align}
where Step \eqref{eq:use_lambda_bound} follows as $\lambda_n\leq 1$. In Step \eqref{eq:def_X2},  $X_{2}^n \sim  \mP_{X_{2}\mid X_1}^{\otimes n}(x_2^n \mid x_1^n)$. Step \eqref{eq:a1_bound} follows from the joint typicality lemma \cite[Sec. 2.5.1]{NIT}, where $\delta_1(\epsilon)\geq 0$ tends to $0$  as $\epsilon \to 0$. 

Similarly, 
\begin{align}
	a_2(x_2^n,y^n) \leq \mP_{X_2}^{\otimes n}(x_2^n) \times 2^{-n(I(X_1;Y\mid X_2)-\delta_2(\epsilon))} \label{eq:a2_bound}
\end{align}
where $\delta_2(\epsilon)\geq 0$ tends to $0$  as $\epsilon \to 0$.

In addition,
\begin{align}
	  a(y^n) & \stackrel{\eqref{eq:def_a12}, \eqref{eq:def_a}}{=} \lambda_n\sum_{x_1^n,x_2^n} \mP_{X_1X_2}^{\otimes n}(x_1^n,x_2^n) \mathbb{I}\big( (x_1^n,x_2^n,y^n) \in \mathcal{T}^{(n)}_\epsilon \big)\\
	&\leq  \Pr\big( (X_1^n,X_2^n,y^n) \in \mathcal{T}^{(n)}_\epsilon \big) \\
	&\leq 2^{-n(I(X_1,X_2; Y) -\delta(\epsilon))}\label{eq:a_bound}
\end{align}
where $(X_1^n,X_2^n)\sim \mP_{X_1X_2}^{\otimes n}$, and $\delta(\epsilon)\geq 0$ tends to $0$ as $\epsilon \to 0$, as implied by the joint typicality lemma.

Going back to \eqref{eq:expansion}, it follows from \eqref{eq:a1_bound}, \eqref{eq:a2_bound} and \eqref{eq:a_bound} that 
\begin{align}
	&\widetilde{\mZ}(x_1^n,x_2^n\mid w_1,w_2,y^n) \\
	&\leq a_{12}(x_1^n,x_2^n,y^n) + \mP_{X_1}^{\otimes n}(x_1^n)\mP_{X_2}^{\otimes n}(x_2^n) \times \notag \\
	&~~~~~~ \big( (M_2-1) 2^{-n(I(X_2;Y\mid X_1)-\delta_1(\epsilon))} \notag\\
	&~~~~~~+(M_1-1) 2^{-n(I(X_1;Y\mid X_2)-\delta_2(\epsilon))} \notag \\
	&~~~~~~+(M_1-1)(M_2-1) 2^{-n(I(X_1,X_2;Y)-\delta(\epsilon))}
	\big) \label{eq:bound_Ztilde}
\end{align}

From now on, restrict $\epsilon >0$ to be sufficiently small so that
\begin{align}
	&I(X_2;Y\mid X_1) - \delta_1(\epsilon) > 0,\\
	&I(X_1;Y\mid X_2) - \delta_2(\epsilon) > 0,\\
	&I(X_1,X_2;Y) - \delta(\epsilon) > 0.
\end{align}
For any $(R_1,R_2)\in \mathbb{R}_{\geq 0}^{2}$ such that
\begin{equation} \label{eq:rate_tuple}
\begin{aligned}
	R_1 &< I(X_1;Y\mid X_2) - \delta_2(\epsilon), \\
	R_2 &< I(X_2;Y\mid X_1) - \delta_1(\epsilon), \\
	R_1+R_2 &< I(X_1,X_2; Y) - \delta(\epsilon),
\end{aligned}
\end{equation}
by picking
\begin{align}
	M_1^{(n)} \triangleq \lfloor 2^{nR_1} \rfloor, ~~~~ M_2^{(n)} \triangleq \lfloor 2^{nR_2} \rfloor,
\end{align}
it is guaranteed that, as $n\to \infty$,
\begin{align}
	\frac{\log_2\big( M_k^{(n)} \big)}{n} \to R_k, ~~~~ \forall k\in \{1,2\}
\end{align}
and that, according to \eqref{eq:bound_Ztilde},
\begin{align}
	&\widetilde{\mZ}(x_1^n,x_2^n\mid w_1,w_2,y^n) \notag \\
	&\leq a_{12}(x_1^n,x_2^n,y^n) + 2^{-\gamma n}  \mP_{X_1}^{\otimes n}(x_1^n)\mP_{X_2}^{\otimes n}(x_2^n) \label{eq:bound_Ztilde_2}
\end{align}
for some $\gamma >0$.

In order to satisfy \eqref{eq:cond_ng}, we choose
\begin{align}
	\lambda_n \triangleq 1-\frac{1}{n},
\end{align}
and
\begin{align}
	\mQ(x_1^n,x_2^n) &\triangleq \lambda_n \mP_{X_1X_2}^{\otimes n} (x_1^n,x_2^n) + (1-\lambda_n) \mP_{X_1}^{\otimes n}(x_1^n)\mP_{X_2}^{\otimes n}(x_2^n)\\
	&=\Big(1-\frac{1}{n}\Big) \mP_{X_1X_2}^{\otimes n} (x_1^n,x_2^n) + \frac{1}{n} \mP_{X_1}^{\otimes n}(x_1^n)\mP_{X_2}^{\otimes n}(x_2^n)
\end{align}
The definition makes $\mQ$ a valid pmf as it is a convex combination of two valid pmfs.
According to the definition in \eqref{eq:def_a12},
\begin{align}
	&a_{12}(x_1^n,x_2^n,y^n) \notag \\
	&\leq \lambda_n \mP_{X_1X_2}^{\otimes n} (x_1^n,x_2^n)\\
	&= \Big(1-\frac{1}{n}\Big) \mP_{X_1X_2}^{\otimes n} (x_1^n,x_2^n)
\end{align}
Together with \eqref{eq:bound_Ztilde_2}, for sufficiently large $n$ (for which $\frac{1}{n}\geq 2^{-\gamma n}$), we  have that
\begin{align}
	\widetilde{\mZ}(x_1^n,x_2^n\mid w_1,w_2,y^n) \leq \mQ(x_1^n,x_2^n),
\end{align}
which meets the desired condition \eqref{eq:cond_ng} for non-negativity. 
Therefore, for every sufficiently small $\epsilon$, the above construction yields valid NS-assisted coding schemes $\mZ^{(n)}$ for all sufficiently large $n$.

\subsection{Reliability analysis}
We now analyze the probability of error of the scheme described above. 
Let us calculate the probability of both messages being decoded correctly. For sufficiently large $n$,
\begin{align}
	&\Pr(\widehat{W}_1=W_1, \widehat{W}_2=W_2) \notag \\
	&\stackrel{\eqref{eq:joint_dist}}{=} \frac{1}{M_1M_2}\sum_{\substack{w_1,w_2\\ x_1^n,x_2^n,y^n}} \mZ(x_1^n,x_2^n,(w_1,w_2)\mid w_1,w_2, y^n) \times \mN^{\otimes n}_{Y\mid X_1X_2}(y^n\mid x_1^n,x_2^n) \\
	& \stackrel{\eqref{eq:def_Z}}{\geq} \frac{1}{M_1M_2} \sum_{\substack{w_1,w_2\\ x_1^n,x_2^n,y^n}} \widetilde{\mZ}(x_1^n,x_2^n,(w_1,w_2)\mid w_1,w_2, y^n) \times \mN^{\otimes n}_{Y\mid X_1X_2}(y^n\mid x_1^n,x_2^n)\\
	& \stackrel{\eqref{eq:def_Ztilde}}{=} \sum_{x_1^n,x_2^n,y^n}a_{12}(x_1^n,x_2^n,y^n) \times \mN^{\otimes n}_{Y\mid X_1X_2}(y^n\mid x_1^n,x_2^n) \\
	& \stackrel{\eqref{eq:def_a12}}{=} \Big(1-\frac{1}{n}\Big) \Pr\big( (X_1^n,X_2^n,Y^n) \in \mathcal{T}_{\epsilon}^{(n)} \big)
\end{align}
where $(X_1^n,X_2^n,Y^n)\sim \mP_{X_1X_2}^{\otimes n}\mN_{Y\mid X_1X_2}^{\otimes n}$. The law of large numbers implies that $\Pr\big( (X_1^n,X_2^n,Y^n) \in \mathcal{T}_{\epsilon}^{(n)} \big) \to 1$ as $n\to \infty$. Therefore, the probability of decoding both messages correctly $\to 1$ as $n\to \infty$. It follows that $\lim_{n \to \infty} P_{e,k}(\mZ^{(n)}) = 0, \forall k\in \{1,2\}$. The rate tuples $(R_1,R_2)$ thus achieved are exactly described by \eqref{eq:rate_tuple}. Taking $\epsilon \to 0$ concludes the proof. \hfil \qed

\section{Conclusion}
The complete characterization of the NS-assisted capacity regions for two of the canonical building blocks of network information theory --- the BC in \cite{Yao_Jafar_NSSato} and now the MAC in this work, highlights the tractability of a NS-extension to classical network information theory, and bodes well for the  development of a comprehensive theory along these lines. The extremal value $K$, of the multiplicative gain in capacity due to NS-assistance across all $K$-user MACs and also across all $K$-user BCs, hints at the intriguing potential residing in non-local correlations to help significantly surpass the classical capacity limits, further underscoring the need to develop a NS-assisted network information theory. Among the questions that are prompted by this work, the NS-assisted MAC capacity region remains open when the NS-assistance is available only to the transmitters. In particular, while a multi-letter characterization of the capacity region is known for the two-sender setting \cite[Proposition 2]{QMAC_Yun}, a single-letter characterization is not known. Another related open question is to determine the NS-assisted capacity region for the MAC \emph{with state} when causal or non-causal CSIT is available to the transmitters. The question is especially intriguing for causal CSIT in light of the discovery in \cite{Yao_Jafar_Unlock} of exponential/unbounded capacity gains from quantum-entanglement assistance in this setting. The prospect of large quantum-assisted capacity gains in the MAC without state also remains open.

\appendix
\section{Proof of Theorem \ref{thm:NSMAC}: Achievability} \label{proof:achievability}
We generalize the proof in Section \ref{proof:two_sender} to the $K$-sender case.
We will prove that, for each $\mP_{X_1X_2\cdots X_K} \in \mathcal{P}(\mathcal{X}_1\times \mathcal{X}_2\times \cdots \times \mathcal{X}_K)$, $(R_1,R_2,\ldots, R_K) \in \mathbb{R}_{\geq 0}^K$ is achievable by NS-assisted coding schemes if
\begin{align}
	\sum_{k\in \mathcal{K}} R_k < I(X_{\mathcal{K}};Y\mid X_{[K]\setminus \mathcal{K}}), ~~~~ \forall \emptyset \neq \mathcal{K} \subseteq [K]
\end{align} 
whenever $\min_{k\in [K]}I(X_k;Y\mid X_{[K]\setminus \{k\}}) >0$. This avoids degenerate cases with fewer senders. Together with time-sharing, this proves the achievability of Theorem \ref{thm:NSMAC}. 
Henceforth, $\mP_{X_1X_2\cdots X_K}$ is fixed to a pmf in $\mathcal{P}(\mathcal{X}_1\times \mathcal{X}_2 \times \cdots \times \mathcal{X}_K)$.

\subsection{A useful sub-distribution}
Let $\epsilon \in (0,1)$. 
In the following, we write $\mathcal{T}^{(n)}_\epsilon(\mP_{X_1X_2\cdots X_K}\mN_{Y\mid X_1X_2\cdots X_K})$ simply as $\mathcal{T}^{(n)}_\epsilon$.
For each $(x_1^n,x_2^n,\ldots, x_K^n,y^n) \in \mathcal{X}_1^n\times \mathcal{X}_2^n\times \cdots \times \mathcal{X}_K^n \times \mathcal{Y}^n$, let
\begin{align} \label{eq:def_a12K}
	&a_{[K]}(x_1^n,x_2^n,\ldots, x_K^n,y^n) \notag \\
	&\triangleq \Big(1-\frac{1}{n}\Big) \times \mP_{X_1X_2\cdots X_K}^{\otimes n}(x_1^n,x_2^n,\ldots, x_K^n) \notag \\
	&~~~~~~~~\times \mathbb{I}\big( (x_1^n,x_2^n,\ldots, x_K^n,y^n) \in \mathcal{T}_{\epsilon}^{(n)} \big).
\end{align}

For each proper subset $\mathcal{S} \subset [K]$, let
\begin{align} \label{eq:def_a_set}
	a_{\mathcal{S}}(x_{\mathcal{S}}^n,y^n) \triangleq \sum_{x^n_{[K]\setminus \mathcal{S}}} a_{[K]}(x_1^n,x_2^n,\ldots, x_K^n,y^n),
\end{align}
where $x_{\mathcal{S}}^n$ is a compact notation for $(x_k^n)_{k\in \mathcal{S}}$.
In particular,
\begin{align}
	a_{\emptyset}(y^n) \triangleq \sum_{x_1^n, x_2^n,\ldots, x_K^n} a_{[K]}(x_1^n,x_2^n,\ldots, x_K^n,y^n).
\end{align}

Then, define an unnormalized distribution
\begin{align} \label{eq:def_Ztilde_K}
	&\widetilde{\mZ}(x_1^n,\ldots, x_K^n, (\widehat{w}_1,\ldots, \widehat{w}_K) \mid w_1,\ldots, w_K, y^n)\notag \\
	&\triangleq 
	\begin{cases}
	a_{[K]}, & \mbox{if}~ (\widehat{w}_k = w_k)_{k\in [K]}\\
	~~~~~~~~~~~~\vdots & ~~~~~~~~~~~~~~~~~~~~ \vdots \\
		a_{[K]\setminus \mathcal{K}} \prod_{j\in \mathcal{K}}\mP_{X_{j}}^{\otimes n} , & \mbox{if}~ (\widehat{w}_k=w_k)_{k\in [K]\setminus \mathcal{K}}, (\widehat{w}_{j} \neq w_{j})_{j\in  \mathcal{K}}\\
	~~~~~~~~~~~~\vdots & ~~~~~~~~~~~~~~~~~~~~ \vdots\\
	a_{\emptyset} \prod_{j\in [K]} \mP_{X_j}^{\otimes n}, & \mbox{if}~ (\widehat{w}_j \neq w_j)_{j\in [K]}
	\end{cases}
\end{align}
where, to avoid cumbersome notation, the parameters $(x_1^n,\ldots, x_K^n,y^n)$ are suppressed in $a_{[K]}$, the parameters $(x_{[K]\setminus \mathcal{K}},y^n)$ are suppressed in $a_{[K]\setminus \mathcal{K}}$, the parameter $(y^n)$ is suppressed in $a_{\emptyset}$, and $(x_j^n)$ is suppressed in $\mP_{X_j}^{\otimes n}$.

We claim that $\widetilde{\mZ}$ satisfies \textit{C1} through \textit{CK}. Due to symmetry it suffices to prove it for \textit{C1} and for the case $(\widehat{w}_2=w_2,\cdots, \widehat{w}_{\ell} = w_{\ell}, \widehat{w}_{\ell+1} \neq w_{\ell+1},\ldots, \widehat{w}_K \neq w_K)$. This is done as follows.
\begin{align}
	&\widetilde{\mZ}(x_2^n,\ldots, x_K^n,(\widehat{w}_1,\ldots, \widehat{w}_K)\mid w_1,\ldots, w_K,y^n) \notag\\
	&\triangleq \sum_{x_1^n} \widetilde{\mZ}(x_1^n,x_2^n,\ldots, x_K^n,(\widehat{w}_1,\ldots, \widehat{w}_K)\mid w_1,\ldots, w_K,y^n)\\
	&= a_{\{2,\ldots, \ell\}}(x_2^n,\ldots, x_{\ell}^n,y^n) \prod_{j=\ell+1}^K \mP_{X_j}^{\otimes n}(x_j^n)
\end{align}
regardless of whether $\widehat{w}_1=w_1$ or $\widehat{w}_1\neq w_1$, and thus it is invariant under changes of $w_1$.

 \subsection{Normalizing the distribution}
Let
\begin{align} \label{eq:def_Z_K}
	&\mZ(x_1^n,\ldots,x_K^n,(\widehat{w}_1,\ldots,\widehat{w}_K)\mid w_1,\ldots,w_K,y^n) \notag \\
	&\triangleq \widetilde{\mZ}(x_1^n,\ldots,x_K^n,(\widehat{w}_1,\ldots,\widehat{w}_K)\mid w_1,\ldots,w_K,y^n) \notag \\
	&~~~~+ \frac{\mQ(x_1^n,\ldots,x_K^n)- \widetilde{\mZ}(x_1^n,\ldots,x_K^n\mid w_1,\ldots,w_K,y^n)}{M_1M_2\cdots M_K} 
\end{align}
where $\mQ\in \mathcal{P}(\mathcal{X}_1^n \times \cdots \times \mathcal{X}_K^n)$ is a valid pmf to be chosen later, and $\widetilde{\mZ}(x_1^n,\ldots,x_K^n\mid w_1,\ldots,w_K,y^n)$ is the marginal (unnormalized) distribution obtained by  summing over all $(\widehat{w}_1,\ldots,\widehat{w}_K)$ in $\widetilde{\mZ}(x_1^n,\ldots,x_K^n,(\widehat{w}_1,\ldots,\widehat{w}_K)\mid w_1,\ldots,w_K,y^n)$.
The form of $\mZ$ in \eqref{eq:def_Z_K} then satisfies \textit{C0}. To see this, note that 
\begin{align}
	&\sum_{\widehat{w}_1,\ldots, \widehat{w}_K} \mZ(x_1^n,\ldots,x_K^n,(\widehat{w}_1,\ldots,\widehat{w}_K)\mid w_1,\ldots,w_K,y^n) \notag \\
	&= \widetilde{\mZ}(x_1^n,\ldots,x_K^n \mid w_1,\ldots,w_K,y^n) + \mQ(x_1^n,\ldots,x_K^n) - \widetilde{\mZ}(x_1^n,\ldots,x_K^n \mid w_1,\ldots,w_K,y^n) \\
	&= \mQ(x_1^n,\ldots,x_K^n)
\end{align}
which is invariant under changes of $y^n$ (and also $w_1,\ldots,w_K$). It can be readily verified that $\mZ$ also satisfies \textit{C1} through \textit{CK}, based on the fact that $\widetilde{\mZ}$ satisfies \textit{C1} through \textit{CK}. Meanwhile, $\mZ$ satisfies the normalization condition, that $\sum_{x_1^n,\ldots,x_K^n, \widehat{w}_1,\ldots, \widehat{w}_K} \mZ(x_1^n,\ldots,x_K^n,(\widehat{w}_1,\ldots,\widehat{w}_K)\mid w_1,\ldots,w_K,y^n) = \sum_{x_1^n,\ldots,x_K^n} \mQ(x_1^n,\ldots,x_K^n) = 1$. 

\subsection{Validity of the solution}
Note that $\mZ$ satisfies \textit{C0},\textit{C1}, $\ldots$, \textit{CK} and the normalization condition. For $\mZ$ to be a valid pmf (and thus an NS scheme), it remains to satisfy the non-negativity condition, that $$\mZ(x_1^n,\ldots,x_K^n,(\widehat{w}_1,\ldots,\widehat{w}_K)\mid w_1,\ldots,w_K,y^n)\geq 0.$$ According to \eqref{eq:def_Z_K}, it suffices to have 
\begin{align} \label{eq:cond_ng_K}
	 \widetilde{\mZ}(x_1^n,\ldots,x_K^n\mid w_1,\ldots,w_K,y^n) \leq \mQ(x_1^n,\ldots,x_K^n) 
\end{align}
for all $x_1^n,\ldots,x_K^n, w_1,\ldots,w_K,y^n$.

To this end, let us first observe that
\begin{align}
& \widetilde{\mZ}(x_1^n,\ldots,x_K^n\mid w_1,\ldots,w_K,y^n) \notag \\
	&\triangleq \sum_{\widehat{w}_1, \ldots, \widehat{w}_K} \widetilde{\mZ}(x_1^n,\ldots,x_K^n,(\widehat{w}_1,\ldots,\widehat{w}_K) \mid w_1,\ldots,w_K,y^n) \\
	&\stackrel{\eqref{eq:def_Ztilde_K}}{=} a_{[K]}(x_1^n,\ldots, x_K^n,y^n) \notag \\
	&~~~~+ \sum_{\emptyset \neq \mathcal{K} \subseteq [K]} a_{[K]\setminus \mathcal{K}}(x_{[K]\setminus \mathcal{K}}^n,y^n) \prod_{j\in \mathcal{K}} (M_j-1)  \mP_{X_j}^{\otimes n}(x_j^n).  \label{eq:expansion_K}
\end{align}
Let us then bound $a_{[K]\setminus \mathcal{K}}(x_{[K]\setminus \mathcal{K}}^n,y^n)$ for each non-empty $\mathcal{K}\subseteq [K]$ as follows.
\begin{align}
	&a_{[K]\setminus \mathcal{K}}(x_{[K]\setminus \mathcal{K}}^n,y^n) \notag \\
	&\leq  \sum_{x_{\mathcal{K}}^n} \mP_{X_1\cdots X_K}^{\otimes n}(x_1^n,\ldots, x_K^n) \notag \\
	&\hspace{2cm} \times \mathbb{I}\big( (x_1^n,\ldots, x_K^n,y^n) \in \mathcal{T}_{\epsilon}^{(n)} \big) \label{eq:use_lambda_bound_K} \\
	&=\mP_{X_{[K]\setminus \mathcal{K}}}^{\otimes n}(x_{[K]\setminus \mathcal{K}}^n) \times   \sum_{x_{\mathcal{K}}^n}  \mP^{\otimes n}_{X_{\mathcal{K}}\mid X_{[K]\setminus \mathcal{K}}}(x_{\mathcal{K}}^n \mid x_{[K]\setminus \mathcal{K}}^n) \notag \\
	& \hspace{2cm} \times \mathbb{I}\big( (x_1^n,\ldots, x_K^n,y^n) \in \mathcal{T}_{\epsilon}^{(n)} \big)\\
	&\leq \mP_{X_{[K]\setminus \mathcal{K}}}^{\otimes n}(x_{[K]\setminus \mathcal{K}}^n) \times 2^{-n(I(X_{\mathcal{K}};Y\mid X_{[K]\setminus \mathcal{K}})-\delta_{[K]\setminus \mathcal{K}}(\epsilon))}. \label{eq:use_jt_K}
\end{align}
Step \eqref{eq:use_lambda_bound_K} follows from the definitions \eqref{eq:def_a12K} and \eqref{eq:def_a_set}. Step \eqref{eq:use_jt_K} follows from the joint typicality lemma, where $\delta_{[K]\setminus \mathcal{K}}(\epsilon)\geq 0$ tends to $0$ as $\epsilon \to 0$.

Going back to \eqref{eq:expansion_K}, we have
\begin{align}
	&\widetilde{\mZ}(x_1^n,\ldots,x_K^n\mid w_1,\ldots,w_K,y^n) \\
	&\leq a_{[K]}(x_1^n,\ldots, x_K^n,y^n) \notag \\
	&~~~~+ \sum_{\emptyset \neq \mathcal{K} \subseteq [K]} \mP_{X_{[K]\setminus \mathcal{K}}}^{\otimes n}(x_{[K]\setminus \mathcal{K}}^n) \times 2^{-n(I(X_{\mathcal{K}};Y\mid X_{[K]\setminus \mathcal{K}})-\delta_{[K]\setminus \mathcal{K}}(\epsilon))}\notag \\
	&~~~~~~~~~~~~~~~~~~~~~~~~~~~~~~ \times  \prod_{j\in \mathcal{K}} (M_j-1)  \mP_{X_j}^{\otimes n}(x_j^n)
	 \label{eq:bound_Ztilde_K}
\end{align}

From now on, restrict $\epsilon >0$ to be sufficiently small so that
\begin{align}
	&I(X_{\mathcal{K}};Y\mid X_{[K]\setminus \mathcal{K}})-\delta_{[K]\setminus \mathcal{K}}(\epsilon) > 0, ~~~~\forall \emptyset \neq \mathcal{K} \subseteq [K].
\end{align}
For any $(R_1,R_2,\ldots, R_K) \in \mathbb{R}_{\geq 0}^K$ such that
\begin{align} \label{eq:rate_tuple_K}
	\sum_{k \in \mathcal{K}} R_k < I(X_{\mathcal{K}}; Y\mid X_{[K]\setminus \mathcal{K}}) - \delta_{[K]\setminus \mathcal{K}}(\epsilon), ~~~~ \forall \emptyset \neq \mathcal{K} \subseteq [K],
\end{align}
by picking
\begin{align}
	M_k^{(n)} \triangleq \lfloor 2^{nR_k} \rfloor, ~~~~ \forall k\in  [K],
\end{align}
it is guaranteed that, as $n\to \infty$,
\begin{align}
	\frac{\log_2\big( M_k^{(n)} \big)}{n} \to R_k, ~~~~ \forall k\in [K] 
\end{align}
and that, 
\begin{align}
	&2^{-n(I(X_{\mathcal{K}};Y\mid X_{[K]\setminus \mathcal{K}})-\delta_{[K]\setminus \mathcal{K}}(\epsilon))} \prod_{j\in \mathcal{K}} (M_j-1) \notag \\
	&\leq 2^{-n(I(X_{\mathcal{K}};Y\mid X_{[K]\setminus \mathcal{K}})-\delta_{[K]\setminus \mathcal{K}}(\epsilon)-\sum_{k \in \mathcal{K}} R_k)} \\
	&\leq 2^{-n\gamma_{\mathcal{K}}}
\end{align}
for some $\gamma_{\mathcal{K}}>0$.
It follows from \eqref{eq:bound_Ztilde_K} that,
\begin{align}
	&\widetilde{\mZ}(x_1^n,\ldots,x_K^n\mid w_1,\ldots,w_K,y^n) \notag \\
	&\leq a_{[K]}(x_1^n,\ldots, x_K^n,y^n) \notag \\
	&~~~~+ \sum_{\emptyset \neq \mathcal{K} \subseteq [K]} \mP_{X_{[K]\setminus \mathcal{K}}}^{\otimes n}(x_{[K]\setminus \mathcal{K}}^n) \prod_{j\in \mathcal{K}} \mP_{X_j}^{\otimes n}(x_j^n)  \times 2^{-n\gamma_{\mathcal{K}}} 
	\label{eq:bound_Ztilde_2_K}
\end{align}

In order to satisfy \eqref{eq:cond_ng_K}, let
\begin{align}
	&\mQ(x_1^n,x_2^n,\ldots, x_K^n) \notag \\
	&\triangleq \Big(1-\frac{1}{n}\Big) \mP_{X_1X_2\cdots X_K}^{\otimes n}(x_1^n,x_2^n,\ldots, x_K^n) \\
	&~~~~+\frac{1/(2^{K}-1)}{n}\sum_{\emptyset \neq \mathcal{K} \subseteq [K]}  \mP_{X_{[K]\setminus \mathcal{K}}}^{\otimes n}(x_{[K]\setminus \mathcal{K}}^n) \prod_{j\in \mathcal{K}} \mP_{X_j}^{\otimes n}(x_j^n) 
\end{align}
which is a valid pmf in $\mathcal{P}(\mathcal{X}_1^n\times \mathcal{X}_2^n \times \cdots \times \mathcal{X}_K^n)$ as it is a convex combination of $2^K$ valid pmfs.
Due to \eqref{eq:def_a12K},
\begin{align}
	a_{[K]}(x_1^n,\ldots, x_K^n,y^n) \leq \Big(1-\frac{1}{n}\Big) \mP_{X_1X_2\cdots X_K}^{\otimes n}(x_1^n,x_2^n,\ldots, x_K^n).
\end{align}
Therefore, for sufficiently large $n$ (for which $\frac{1/(2^K-1)}{n} \geq 2^{-n\gamma_{\mathcal{K}}}, \forall \emptyset \neq \mathcal{K} \subseteq [K]$), we have $\widetilde{\mZ}(x_1^n,\ldots,x_K^n\mid w_1,\ldots,w_K,y^n)  \leq \mQ(x_1^n,\ldots,x_K^n)$, which meets the desired condition \eqref{eq:cond_ng_K} for non-negativity. Therefore, for every sufficiently small $\epsilon$, the above construction yields valid NS-assisted coding schemes $\mZ^{(n)}$ for all sufficiently large $n$.
\subsection{Reliability analysis}
Let us calculate the probability of all $K$ messages being decoded correctly. 
For sufficiently large $n$,
\begin{align}
	&\Pr(\widehat{W}_k=W_k, \forall k\in [K]) \notag \\
	&\stackrel{\eqref{eq:joint_dist}}{=} \frac{1}{M_1M_2\cdots M_K}\sum_{\substack{w_1,\ldots, w_K\\ x_1^n,\ldots,x_K^n,y^n}} \mZ(x_1^n,\ldots,x_K^n,(w_1,\ldots,w_K)\mid w_1,\ldots,w_K, y^n) \notag \\
	&\hspace{6cm} \times \mN^{\otimes n}_{Y\mid X_1\cdots X_K}(y^n\mid x_1^n,\ldots,x_K^n) \\
	&\stackrel{\eqref{eq:def_Z_K}}{\geq} \frac{1}{M_1M_2\cdots M_K} \sum_{\substack{w_1,\ldots,w_K\\ x_1^n,\ldots,x_K^n,y^n}} \widetilde{\mZ}(x_1^n,\ldots,x_K^n,(w_1,\ldots,w_K)\mid w_1,\ldots,w_K, y^n) \notag \\
	&\hspace{6cm} \times \mN^{\otimes n}_{Y\mid X_1\cdots X_K}(y^n\mid x_1^n,\ldots, x_K^n)\\
	& \stackrel{\eqref{eq:def_Ztilde_K}}{=} \sum_{x_1^n,\ldots,x_K^n,y^n}a_{[K]}(x_1^n,\ldots,x_K^n,y^n) \times \mN^{\otimes n}_{Y\mid X_1\cdots X_K}(y^n\mid x_1^n,\ldots,x_K^n) \\
	& \stackrel{\eqref{eq:def_a12K}}{=} \Big(1-\frac{1}{n}\Big) \Pr\big( (X_1^n,\ldots,X_K^n,Y^n) \in \mathcal{T}_{\epsilon}^{(n)} \big)
\end{align}
where $(X_1^n,\ldots,X_K^n,Y^n)\sim \mP_{X_1\cdots X_K}^{\otimes n}\mN_{Y\mid X_1\cdots X_K}^{\otimes n}$. The law of large numbers implies that $$\Pr\big( (X_1^n,\ldots,X_K^n,Y^n) \in \mathcal{T}_{\epsilon}^{(n)} \big) \to 1$$ as $n\to \infty$. Therefore, the probability of decoding all $K$ messages correctly $\to 1$ as $n\to \infty$. It follows that $\lim_{n \to \infty} P_{e,k}(\mZ^{(n)}) = 0, \forall k\in [K]$. The rate tuples $(R_1,R_2,\ldots, R_K)$ thus achieved are exactly described by \eqref{eq:rate_tuple_K}. Taking $\epsilon \to 0$ concludes the proof of achievability. \hfil \qed

\section{Proof of Theorem \ref{thm:NSMAC}: Converse} \label{proof:converse}
Let $(R_1,R_2,\ldots, R_K)$ be any achievable rate tuple, and let $(\mZ^{(n)})_{n\in \mathbb{N}}$ be a sequence of $$(M_1^{(n)},M_2^{(n)}, \ldots, M_K^{(n)}, n)$$ NS-assisted coding schemes for which $\liminf_{n\to \infty}\log_2(M_k^{(n)})/n = R_k$ and $\lim_{n\to \infty} P_{e,k}(\mZ^{(n)}) = 0, \forall k\in [K]$. Let $\epsilon_n$ be the probability that at least one message is incorrectly decoded  associated with $\mZ^{(n)}$. It follows that $\epsilon_n\to 0$ as $n\to \infty$.

For each non-empty subset $\mathcal{K}\subseteq [K]$, consider the transmitters indexed by $k\in \mathcal{K}$ as one big transmitter (say Transmitter A) and the transmitters indexed by $j\in [K]\setminus \mathcal{K}$ as another big transmitter (say Transmitter B, could be empty). We apply \cite[Lemma 38]{fawzi2024MAC} by treating Transmitter A as the first transmitter and Transmitter B as the second transmitter. It follows that
\begin{align}
	\log_2(\prod_{k \in \mathcal{K}} M_k) \leq \frac{I(X_{\mathcal{K}}^n;Y^n\mid X_{[K]\setminus \mathcal{K}}^n)+H_b(\epsilon_n)}{1-\epsilon_n}
\end{align}
where $H_b(\cdot)$ denotes the binary entropy function.
Since $\log_2(\prod_{k \in \mathcal{K}} M_k) = \sum_{k\in \mathcal{K}} \log_2(M_k) \geq n\sum_{k\in \mathcal{K}}R_k - o(n)$, we have
\begin{align}
	&n\sum_{k\in \mathcal{K}}R_k-o(n) \notag \\
	&\leq I(X_{\mathcal{K}}^n;Y^n\mid X_{[K]\setminus \mathcal{K}}^n) \\
	&\leq \sum_{i=1}^n I(X_{\mathcal{K},i};Y_i\mid X_{[K]\setminus \mathcal{K},i})\label{eq:use_ff_1}\\
	&=n I(X_{\mathcal{K},Q};Y_Q\mid X_{[K]\setminus \mathcal{K},Q} ,Q) \label{eq:time_sharing}
\end{align}
Step \eqref{eq:use_ff_1} applies \cite[Lemma 39]{fawzi2024MAC}. In Step \eqref{eq:time_sharing}, $Q$ is a time-sharing variable, uniformly distributed over $[n]$ and independent of $(X_{[K]}^n,Y^n)$. Taking $n\to \infty$, it shows that for every non-empty $\mathcal{K}\subseteq [K]$, we have $\sum_{k\in \mathcal{K}}R_k \leq I(X_{\mathcal{K}};Y\mid X_{[K]\setminus \mathcal{K}}, Q)$ for $(Q,X_{[K]}, Y)$ distributed according to some $\mP_{QX_1\cdots X_K Y} = \mP_{QX_1\cdots X_K} \mN_{Y\mid X_1\cdots X_K}$. 
Taking $n\to \infty$ establishes the converse without a cardinality restriction on the time-sharing variable $Q$. The cardinality bound $|\mathcal{Q}| \leq K$ then follows from the Fenchel-Eggleston-Carath\'{e}odory theorem (see, e.g., \cite[Appendix 4A]{NIT}). This completes the proof of converse. \hfil \qed

\section{Proof of Theorem \ref{thm:sum_prod_lock}: $C_{\Sigma}^{\NS} \geq K\log_2(q-1)$} \label{proof:sum_prod_lock_NSa}
Let us provide two proofs. 
The first one applies Theorem \ref{thm:NSMAC} as follows.

Let $X_k = (A_k,B_k), \forall k\in [K]$. Let $(A_1,A_2,\ldots, A_{K-1},A_K, B_1,B_2,\ldots, B_{K-1})$ be uniformly distributed over $(\mathbb{F}_q^{\times})^K \times \mathbb{F}_q^{K-1}$. Let $B_K = \prod_{i=1}^K A_i -\sum_{j=1}^{K-1} B_j$, which is uniformly distributed over $\mathbb{F}_q$. It follows that for any $k\in [K]$, the $2K-1$ random variables $(A_1,\ldots, A_K, B_1,\ldots, B_{k-1}, B_{k+1}, \ldots, B_K)$ are independent. Note that with these inputs, the channel never erases, so $Y = (A_1,A_2,\ldots, A_K)$. Therefore, for any non-empty subset $\mathcal{K}\subseteq [K]$, we have $I(X_{\mathcal{K}};Y\mid X_{[K]\setminus \mathcal{K}}) \geq H(A_{\mathcal{K}}\mid A_{[K]\setminus \mathcal{K}},B_{[K]\setminus \mathcal{K}}) = H(A_{\mathcal{K}}) = |\mathcal{K}|\log_2 (q-1)$. The rate tuple $(R_1, \ldots, R_K) = (\log_2 (q-1), \ldots, \log_2 (q-1)) $ is thus in the region $\mathcal{C}^{\NS}$. It follows that $C_{\Sigma}^{\NS}\geq K\log_2(q-1)$. \hfil \qed

The second proof, presented next, needs only transmitters' side NS assistance. 

Consider a $K$-partite NS box $\mB$ shared by the $K$ transmitters. The box is modeled as a conditional pmf $\mB(b_1,\ldots, b_K\mid a_1,\ldots, a_K)\in \mathcal{P}(\mathbb{F}_q^K\mid (\mathbb{F}_q^{\times})^K)$, with inputs denoted by $(A_1,\ldots, A_K)$ and outputs denoted by $(B_1,\cdots, B_K)$, respectively, for the $K$ transmitters. The conditional pmf is defined as,
\begin{align}
	&\mB(b_1,\ldots, b_K\mid a_1,\ldots, a_K)\notag \\
	&=\begin{cases}
		q^{-(K-1)}, & \mbox{if}~ \sum_{j=1}^K b_j = \prod_{i=1}^K a_i\\
		0, & \mbox{otherwise}
	\end{cases}.
\end{align}
Note that, for any $k\in [K]$, $\sum_{b_k\in \mathbb{F}_q}\mB(b_1,\ldots, b_K\mid a_1,\ldots, a_K) = q^{-(K-1)}$, which is invariant under changes of $a_k\in \mathbb{F}_q^{\times}$ (and in fact all $(a_1,\ldots, a_K)$). Thus $\mB$ is non-signaling. Now suppose for each $k\in [K]$, the message $W_k$ is uniformly distributed over $\mathbb{F}_q^{\times}$. Transmitter $k$ inputs  $A_k=W_k$ to the box $\mB$, and the box produces $B_k$. Transmitter $k$ then sends $X_k=(A_k,B_k)$ through the channel. Due to the construction of the box, the $(A_1,\ldots, A_K)$ and $(B_1,\ldots, B_K)$ thus generated always satisfy the condition in \eqref{eq:def_channel_sum_prod_lock} under which the channel produces $(A_1,\ldots, A_K)$ to the receiver. Therefore, the receiver is able to receive all messages $(W_1,\ldots, W_K)$ correctly. The rate tuple $(R_1,\ldots, R_K) = (\log_2(q-1), \ldots, \log_2(q-1))$ is thus achieved without error. It follows that $C_{\Sigma}^{\NS}\geq K\log_2(q-1)$. \hfil \qed

\section{Proof of Theorem \ref{thm:sum_prod_lock}: $C_{\Sigma} \leq \log_2 q+ o_q(\log_2 q)$} \label{proof:sum_prod_lock_Cc}
Let $(X_1,X_2,\ldots, X_K)\sim \mP_{X_1}\mP_{X_2}\cdots \mP_{X_K}$ and $Y$ be the output of the channel. We will prove that $H(Y) \leq \log_2 q+o_q(\log_2 q)$. 

Denote the set of inputs for which the channel does not produce $\phi$, i.e.,
\begin{align}
	\mathcal{E} \triangleq \Bigg\{(x_1,\ldots, x_K)\in \prod_{k=1}^K  \mathcal{X}_k \colon \sum_{j=1}^K b_j = \prod_{i=1}^K a_i \Bigg\},
\end{align}
where $x_i = (a_i,b_i)$. We view $\mathcal{E}$ as (the edge set of) a $K$-partite, $K$-uniform hypergraph\footnote{A $K$-partite, $K$-uniform hypergraph is a hypergraph whose vertex set is partitioned into $K$ disjoint classes, with every hyperedge containing exactly one vertex from each class.} whose $i^{th}$ vertex class is $\mathcal{X}_i$.
In the following, we typically identify a graph (and hypergraph) with its edge set.
The definition of $\mathcal{E}$ ensures that it satisfies the following conditional $\mathcal{K}_{2,2}$-freeness property.
\begin{definition}[Conditional $\mathcal{K}_{2,2}$-freeness] \label{def:cond_K22free}
	Let $\kappa\geq 2$ and $\mathcal{H}\subseteq \mathcal{V}_1\times \cdots \times \mathcal{V}_{\kappa}$ be a $\kappa$-partite, $\kappa$-uniform hypergraph whose $i^{th}$ vertex class is $\mathcal{V}_i$. We say that $\mathcal{H}$ is conditionally $\mathcal{K}_{2,2}$-free if, for every pair of distinct indices $i<j \in [\kappa]$ and every fixed choice of vertices $v_k\in \mathcal{V}_k$ for all $k\in[\kappa]\setminus \{i,j\}$, the resulting bipartite graph
	\begin{align}
		\mathcal{G}_{ij} \triangleq \{(a,b) \in \mathcal{V}_i \times \mathcal{V}_j \colon (v_1, \ldots, v_{i-1}, a, v_{i+1}, \ldots, v_{j-1}, b, v_{j+1}, \ldots, v_{\kappa}) \in \mathcal{H} \}
	\end{align}
	does not contain $\mathcal{K}_{2,2}$ (the complete bipartite graph on two vertices in each partition class).
\end{definition}
The following lemma follows directly from the definition of $\mathcal{E}$.
\begin{lemma} \label{lem:K22free}
	For $K\geq 2$, $\mathcal{E}$ is conditionally $\mathcal{K}_{2,2}$-free.
\end{lemma}
\begin{proof}
	Assume to the contrary that, without loss of generality, the following four hyperedges
	\begin{align*}
		(x_1,x_2,\ldots, x_{K-1}, x_K) \\
		(x_1',x_2,\ldots, x_{K-1}, x_K)\\
		(x_1,x_2,\ldots, x_{K-1}, x_K')\\
		(x_1',x_2,\ldots, x_{K-1}, x_K')
	\end{align*} 
	are in $\mathcal{E}$. Recall the notation $x_i=(a_i,b_i) \in \mathbb{F}_q^{\times}\times \mathbb{F}_q$. Consider that $v,v'$ are points on the finite plane $\mathbb{F}_q^2$ where $v \triangleq (\prod_{i=1}^{K-1}a_i,\sum_{j=1}^{K-1}b_j)$, $v' \triangleq (a_1'\prod_{i=2}^{K-1}a_i,b_1'\sum_{j=2}^{K-1}b_j)$, and that $\ell(x_K), \ell(x_K')$ define lines on the finite plane, where $\ell(x_K)\triangleq \{(x,y)\in \mathbb{F}_q^2\colon y=a_Kx-b_K\}$, $\ell(x_K')\triangleq \{(x,y)\in \mathbb{F}_q^2\colon y=a_K'x-b_K'\}$. Then $v$ and $v'$ are distinct points, and $\ell(x_K)$ and $\ell(x_K')$ are distinct lines, on the finite plane $\mathbb{F}_q^2$. However, the initial assumption implies that both $v$ and $v'$ are incident with both $\ell(x_K)$ and $\ell(x_K')$. This contradicts the fact that two distinct lines can intersect in at most one point. 
\end{proof}
The following lemma is a consequence of the conditional $\mathcal{K}_{2,2}$-free property.
\begin{lemma} \label{lem:K22pK}
	Let $\mathcal{H}\subseteq \mathcal{V}_1\times \cdots \times \mathcal{V}_{\kappa}$ be a $\kappa$-partite, $\kappa$-uniform hypergraph whose $i^{th}$ vertex class is $\mathcal{V}_i$. Let $|\mathcal{V}_i| = n_i$, $|\mathcal{H}| = e$, and assume $e>0$. Define the edge density $\rho \triangleq \frac{e}{n_1n_2\cdots n_{\kappa}}$. Then, there exists an index $k^*\in [\kappa]$, such that
	\begin{align}
		n_k \leq \frac{4}{\rho^2}, ~~~~ \forall k\in [\kappa] \setminus \{k^*\}.
	\end{align}
	In other words, at most one of $n_1,n_2,\ldots, n_{\kappa}$ can be bigger than $\frac{4}{\rho^2}$.
\end{lemma}
To prove Lemma \ref{lem:K22pK}, we need the following intermediate lemma.
\begin{lemma}[KST] \label{lem:K22p2}
	Let $\mathcal{G}\subseteq \mathcal{U} \times \mathcal{V}$ be a bipartite graph that does not contain $\mathcal{K}_{2,2}$. Let $|\mathcal{U}|=m$, $|\mathcal{V}| = n$, and $|\mathcal{G}| = e$. Then
	\begin{align}
		e \leq n+m\sqrt{n}.
	\end{align}
	By symmetry, we also have $e\leq m+n\sqrt{m}$.
\end{lemma}
\begin{proof}
	It is a special case of the K\H{o}v\'ari--S\'os--Tur\'an Theorem \cite{kovari1954problem} (see also the introduction of \cite{K22free}). For completeness, we provide a proof in Appendix \ref{proof:lem_K22p2}.
\end{proof}
Let us prove Lemma \ref{lem:K22pK}. 
\begin{proof}[Proof of Lemma \ref{lem:K22pK}]
	Fix distinct indices $i,j\in [\kappa]$. For every fixed choice $c=(v_k \in \mathcal{V}_k)_{k\in [\kappa]\setminus\{i,j\}}$ of vertices, let $e_c$ denote the number of hyperedges in $\mathcal{H}$ whose coordinates outside $i,j$ are $c$. Since $\mathcal{H}$ is conditionally $\mathcal{K}_{2,2}$-free, Lemma \ref{lem:K22p2} implies that $e_c \leq n_j + n_i \sqrt{n_j}$. Summing over all $\prod_{k\in [\kappa]\setminus \{i,j\}}n_k$ choices of $c$ gives
	\begin{align}
		e\leq \Big(\prod_{k \in [\kappa]\setminus \{i,j\}}n_k \Big) (n_j+n_i\sqrt{n_j}).
	\end{align}
	Dividing by $n_1n_2\cdots n_{\kappa}$ yields
	\begin{align}
		\rho \leq \frac{1}{n_i} + \frac{1}{\sqrt{n_j}}.
	\end{align}
	If both $n_i > 2/\rho$ and $n_j> 4/\rho^2$, then we would have $\rho < \rho/2+\rho/2 =\rho$, which is a contradiction. Therefore, we must have $n_i \leq 2/\rho$ or $n_j \leq 4/\rho^2$. Since $2/\rho \leq 4/\rho^2$, we have that $\min\{n_i, n_j\} \leq 4/\rho^2$. Since the argument works for every distinct pair of $\{i,j\}$, it follows that at most one of $n_1,n_2,\ldots, n_{\kappa}$ can be bigger than $4/\rho^2$. We thus conclude the proof.
\end{proof}

With Lemma \ref{lem:K22pK} we will proceed to prove the entropic bound $H(Y)\leq \log_2 q +o_q(\log_2 q)$. Recall that $(X_1,\ldots, X_K) \sim \mP_{X_1}\cdots \mP_{X_K}$. We first need to partition each set $\mathcal{X}_k$ into subsets according to the probability $\mP_{X_k}(x_k)$ for $x_k\in \mathcal{X}_k$.  
Specifically, for each $k\in [K]$, let
\begin{align}
	\mathcal{S}_{k,*} \triangleq \{x_k\in \mathcal{X}_k\colon \mP_{X_k}(x_k) \leq q^{-3}\}.
\end{align}
Let 
\begin{align}
	\eta \triangleq   \frac{1}{\sqrt{\log_2 q}}.
\end{align}
The motivation for this choice will become clear later. For each integer $0 \leq \ell < 3/\eta$, let
\begin{align}
	\mathcal{S}_{k,\ell} \triangleq \big\{x_k\in \mathcal{X}_k \colon \max\{ q^{-\eta (\ell+1)}, q^{-3} \} < \mP_{X_k}(x_k) \leq q^{-\eta \ell} \big\}.
\end{align}
In other words, all $x_k$ whose probability masses do not exceed $q^{-3}$ are assigned to the level labeled by $*$, while the remaining $x_k$ are grouped into levels according to multiplicative intervals with ratio $q^{-\eta} = 2^{-\sqrt{\log_2 q}}$. 

For each $k\in [K]$, let $S_k$ denote the level index of the random variable $X_k$, i.e., $S_k = \ell$ if $X_k \in \mathcal{S}_{k,\ell}$.
Note that $S_k$ has at most $L = \big\lceil \frac{3}{\eta} \big\rceil + 1$ possible values. Therefore,
\begin{align} \label{eq:entropy_for_classes}
	H(S_1,S_2,\ldots, S_K) \leq K\log(L) = o_q(\log_2 q).
\end{align}
Since for all $k\in [K]$, we have $|\mathcal{X}_k| = q(q-1) < q^2$, and thus $\Pr(S_k = *)\leq q^2 q^{-3} = q^{-1}$. 
By the union bound, 
\begin{align}
	\Pr(\mbox{at least one}~ S_k = *) \leq \frac{K}{q}.
\end{align}
Consider the conditional entropy
\begin{align}
	&H(Y\mid S_1,\ldots, S_K)\notag \\
	&=\sum_{(s_1,\ldots, s_K)} \Pr(S_1=s_1,\ldots, S_K=s_K) \times \notag \\
	&\hspace{2cm} H(Y\mid S_1=s_1,\ldots, S_K=s_K)
\end{align}
Since $|\mathcal{Y}|=(q-1)^K+1$, the total contribution to the conditional entropy $H(Y\mid S_1,S_2,\ldots,S_K)$ from those $(s_1,\ldots,s_K)$ for which at least one $s_k=*$ is at most
\begin{align} \label{eq:negligible_classes}
	\frac{K}{q} \log_2((q-1)^{K}+1) = o_q(1).
\end{align}
Next we bound $H(Y\mid S_1=s_1,\ldots, S_K=s_K)$ for each fixed $(s_1,\ldots, s_K)$ such that $s_k\neq *, \forall k\in [K]$ and that $\Pr(S_k=s_k, \forall k\in [K])>0$.
Let
\begin{align}
	n_k \triangleq |\mathcal{S}_{k,s_k}| > 0.
\end{align}
For every $x_k\in \mathcal{S}_{k,s_k}$, 
\begin{align}
	&\Pr(X_k=x_k\mid S_k=s_k) \\
	&=\frac{\Pr(X_k=x_k)}{\Pr(S_k=s_k)}\\
	&\leq \frac{q^{-\eta s_k}}{n_k q^{-\eta(s_k+1)}}\\
	&= \frac{q^{\eta}}{n_k}
\end{align}
Since $X_1,X_2,\ldots, X_K$ are independent and $S_k$ is a function of $X_k$, we have that for every $(x_1,\ldots, x_K) \in \mathcal{S}_{1,s_1}\times \cdots \times \mathcal{S}_{K,s_K}$, 
\begin{align}
	&\Pr(X_1=x_1,\ldots, X_K=x_K\mid S_1=s_1,\ldots, S_K=s_K)\notag \\
	&= \prod_{k=1}^K \Pr(X_k=x_k\mid S_k=s_k) \\
	&\leq \frac{q^{K\eta}}{\prod_{k=1}^K n_k}
\end{align}
Let
\begin{align}
	e \triangleq |\mathcal{E} \cap (\mathcal{S}_{1,s_1} \times \cdots \times \mathcal{S}_{K,s_K})|,
\end{align}
and
\begin{align}
	\alpha \triangleq \Pr((X_1,\ldots, X_K) \in \mathcal{E} \mid S_1=s_1,\ldots, S_K=s_K).
\end{align}
Then
\begin{align} \label{eq:alpha_is_small}
	\alpha \leq \frac{e q^{K\eta}}{\prod_{k=1}^K n_k}.
\end{align}
Note that $\alpha$ is the conditional probability that the channel does not produce the erasure symbol $\phi$ conditioned on $S_1=s_1,\ldots, S_K=s_K$. 
Now consider two cases.
\begin{enumerate}
	\item If $\alpha \leq q^{-\eta}$, then $H(Y\mid S_1=s_1,\ldots, S_K=s_K) \leq H_b(\alpha) + \alpha \log_2((q-1)^K) \leq  H_b(\alpha) + K2^{-\sqrt{\log_2 q}}\log_2q  = o_q(\log_2 q)$.
	\item If $\alpha > q^{-\eta}$. Then \eqref{eq:alpha_is_small} implies $\frac{e q^{K\eta}}{\prod_{k=1}^K n_k} \geq q^{-\eta}$ and therefore $\frac{e}{\prod_{k=1}^K n_k} \geq q^{-\eta(K+1)}$. Applying Lemma \ref{lem:K22pK} to the sub-hypergraph of $\mathcal{E}$ with $\mathcal{V}_k = \mathcal{S}_{k,s_k}, \forall k\in [K]$, we have that there exists a $k^*\in [K]$, such that
	\begin{align}
		n_{k} \leq 4\Big( \frac{e}{\prod_{k=1}^K n_k} \Big)^{-2} \leq 4q^{2\eta (K+1)}, ~~ \forall k\in [K]\setminus \{k^*\}.
	\end{align}
	In other words, all $n_k$ except for $n_{k^*}$ satisfy the upper bound $4q^{2\eta (K+1)}$. The value of $k^*$ may be chosen as $1$ if all $n_k$ are upper bounded by $4q^{2\eta (K+1)}$. We then have
	\begin{align}
		&H(X_{[K]\setminus \{k^*\}}\mid S_1=s_1,\ldots, S_K=s_K)\notag \\
		&\leq \sum_{k\in [K]\setminus \{k^*\}} H(X_k\mid S_1=s_1,\ldots, S_K=s_K) \\
		&\leq \sum_{k\in [K]\setminus \{k^*\}}\log_2 n_k \\
		&\leq (K-1)(2\eta (K+1) \log_2 q+2)\\
		&=O_q(\sqrt{\log_2 q})\\
		&=o_q(\log_2 q).
	\end{align}
	Therefore,
	\begin{align}
		&H(Y\mid S_1=s_1,\ldots, S_K=s_K)\notag \\
		&\leq H(Y, X_{[K]\setminus \{k^*\}}\mid S_1=s_1,\ldots, S_K=s_K)\\
		&=H(X_{[K]\setminus \{k^*\}}\mid S_1=s_1,\ldots, S_K=s_K)\notag \\
		&~~~~+H(Y\mid X_{[K]\setminus \{k^*\}}, S_1=s_1,\ldots, S_K=s_K)\\
		&\leq o_q(\log_2 q) + \log_2 q
	\end{align}
	The last step follows because, once $X_{k}=(A_k,B_k)$ is fixed for every $k\in [K]\setminus \{k^*\}$, the only remaining freedom in $Y$ is $A_{k^*}\in \mathbb{F}_q^{\times}$, plus the erasure symbol.
\end{enumerate}
Therefore in both cases we obtain $H(Y\mid S_1=s_1,\ldots, S_K=s_K) \leq \log_2 q+ o_q(\log_2 q)$. Together with \eqref{eq:negligible_classes}, taking the average,
\begin{align}
	&H(Y\mid S_1,\ldots, S_K)\notag \\
	&\leq \log_2 q + o_q(\log_2 q).
\end{align}
Combining it with \eqref{eq:entropy_for_classes}, we have that
\begin{align}
	H(Y) \leq H(S_1,\ldots, S_K) + H(Y\mid S_1,\ldots, S_K) \leq \log_2 q + o_q(\log_2 q).
\end{align} 
Note that this holds for every product distribution $\mP_{X_1}\cdots \mP_{X_K}$. Therefore, we conclude that $C_{\Sigma} \leq \log_2 q + o_q(\log_2 q)$.
\hfil \qed

\section{Proof of Lemma \ref{lem:K22p2}} \label{proof:lem_K22p2}
For $v \in \mathcal{V}$, let $d(v) \triangleq |\{u \in \mathcal{U} \colon (u,v) \in \mathcal{G}\}|$ be its degree (into $\mathcal{U}$). Then
\begin{align}
	\sum_{v \in \mathcal{V}} d(v) = e. \label{eq:sum_dv_and_e}
\end{align}
For any element $v \in \mathcal{V}$, $\binom{d(v)}{2}$ counts the number of (unordered) pairs of  $\{u_1,u_2\} \subseteq \mathcal{U}$ that have $v$ as a common neighbor.  
Note that $\binom{m}{2}$ counts the total number of unordered pairs of vertices $\{u_1,u_2\} \subseteq \mathcal{U}$. Since  $\mathcal{G}$ contains no $\mathcal{K}_{2,2}$, any (unordered) pair of vertices $\{u_1,u_2\} \subseteq \mathcal{U}$ can have at most one common neighbor in $\mathcal{V}$. Therefore, 
\begin{align} \label{eq:use_K22}
	\binom{m}{2} &\geq \sum_{v \in \mathcal{V}} \binom{d(v)}{2}\\
	&\geq n \binom{\sum_{v\in \mathcal{V}}d(v)/n}{2} \label{eq:use_convex} \\
	&= n \binom{e/n}{2} \label{eq:use_sum_dv_and_e}
\end{align}
where Step \eqref{eq:use_convex} uses convexity of the function $f(x) = \binom{x}{2}=\frac{x(x-1)}{2}$ and Step \eqref{eq:use_sum_dv_and_e} uses \eqref{eq:sum_dv_and_e}. Therefore, $m^2\geq m(m-1) \geq e \big(\frac{e}{n}-1\big)$, and thus $e^2 -ne -nm^2 \leq 0$. It follows that $e$ is upper bounded by the larger root of the quadratic equation $x^2-nx-nm^2=0$, i.e.,
\begin{align}
	e \leq \frac{n+\sqrt{n^2+4nm^2}}{2} \leq \frac{n+(n+2m\sqrt{n})}{2} = n+m\sqrt{n}.
\end{align}
This completes the proof of Lemma \ref{lem:K22p2}. \hfil \qed

\section{A clarification of Table II of Yun et al. \cite{QMAC_Yun}} \label{sec:yunetal}
Reference \cite{QMAC_Yun} studies the sum-rate capacity of two-sender MACs when the two transmitters are provided with non-signaling assistance. 
Table II of \cite{QMAC_Yun} contains several capacity results for a set of two-user MACs parameterized by a scalar $\eta\in [0.5,1]$. In particular, it compares $C_{\Sigma}$, which is the classical sum-rate capacity, with $C_{\Sigma}^{\TxNS}$, which is the sum-rate capacity with transmitter-side NS assistance. 
In the case of $\eta = 0.5$, the table indicates that $C_{\Sigma} = 0.451$ and $C_{\Sigma}^{\TxNS} = 1$. Since all-party NS assistance subsumes transmitter-side NS assistance, this would imply that $C_{\Sigma}^{\NS}\geq 1$, seemingly achieving a multiplicative gain in capacity from NS-assistance that exceeds the extremal value of $K=2$ for this case. We clarify this contradiction here by pointing out a minor oversight in the calculation of $C_{\Sigma}^{\TxNS}$ in \cite{QMAC_Yun}. 

The channel of interest is described as follows. The channel alphabets are $\mathcal{Y}=\mathcal{X}_1=\mathcal{X}_2 = \{0,1\}^2$, and the channel law is that, for every $y\in \{0,1\}^2,(a_1,a_2)\in \{0,1\}^2, (b_1,b_2)\in \{0,1\}^2$,
\begin{align}
	&\Pr(Y=y\mid X_1=(a_1,a_2), X_2= (b_1,b_2)) \notag \\
	&= \begin{cases}
		\eta \mathbb{I}(y=(a_1,b_1))+\frac{1-\eta}{4}, & \mbox{if}~ a_2\oplus b_2 = a_1b_1\\
		(1-\eta)\mathbb{I}(y=(a_1,b_1))+\frac{\eta}{4}, &\mbox{otherwise}
	\end{cases}
\end{align}
where $\oplus$ denotes the addition modulo $2$. The operational meaning of the channel is as follows. If the inputs $(a_1,a_2), (b_1,b_2)$ satisfy the condition $a_2\oplus b_2 = a_1b_1$, then the channel outputs $Y = (a_1,b_1)$ with probability $\eta$, and outputs a uniform noise in $\{0,1\}^2$ with probability $1-\eta$. If the inputs do not satisfy that condition, then the channel outputs $Y = (a_1,b_1)$ with probability $1-\eta$, and outputs a uniform noise in $\{0,1\}^2$ with probability $\eta$. 

In the case of $\eta=0.5$, we note that the channel behaviors turn out to be the same in both cases, and the channel law can be simplified as $\Pr(Y=y\mid X_1=(a_1,a_2), X_2= (b_1,b_2)) = \frac{1}{2}  \mathbb{I}(y=(a_1,b_1))+\frac{1}{8}$. Equivalently, the channel admits the form $Y = \bbsmatrix{a_1\\b_1} \oplus Z$, where $Z\in \{0,1\}^2$  such that $\Pr(Z=\bbsmatrix{0\\0}) = \frac{5}{8}$, and $\Pr(Z=\bbsmatrix{z_1\\z_2}) = \frac{1}{8}$ for all $\bbsmatrix{z_1\\z_2} \neq \bbsmatrix{0\\0}$. 
Even if both transmitters collaborate, the sum-rate capacity cannot be more than $2-H(Z)\approx 0.4512$ (bits), which is the classical capacity for the reduced point-to-point channel. This is because for a point-to-point channel, non-signaling assistance does not improve the capacity \cite{matthews2012linear}. Therefore, the correct value should satisfy $C_{\Sigma}^{\TxNS} \leq C_{\Sigma}^{\NS} \leq 2-H(Z)\approx 0.4512$ for $\eta = 0.5$. We do note, however, that Figure 3 of \cite{QMAC_Yun} seems to be consistent with the corrected result.

 \section*{Acknowledgment}
The authors acknowledge the assistance of ChatGPT 5.6 Sol in this work. Following a series of discussions focused on the binary-adder channel and potential generalizations of the authentication schemes of \cite{Yao_Jafar_NSSato}, ChatGPT 5.6 Sol helped identify the structure of the NS-assisted scheme used to prove the achievability part of Theorem \ref{thm:NSMAC} for the $K=2$ setting. The insights from the $K=2$ setting led the authors to the general achievability proof, which was verified first with the help of ChatGPT 5.6 Sol and then manually by the authors. ChatGPT 5.6 Sol was also used to identify the class of $K$-user MACs that attains the extremal ratio in Theorem \ref{thm:sum_prod_lock} and to assist in developing its proof, particularly Lemmas \ref{lem:K22free}, \ref{lem:K22pK}, and \ref{lem:K22p2}. It also assisted with figure preparation and literature searches. The authors developed and verified the final proofs and accept responsibility for any remaining errors.

\bibliographystyle{IEEEtran}
\bibliography{../../bib_file/yy.bib}
\end{document}